\documentclass[11pt,letterpaper]{article}
\usepackage{float}
\usepackage[margin=1in]{geometry}  
\usepackage[T1]{fontenc}        
\usepackage{newtxtext,newtxmath}

\usepackage{authblk}
\usepackage{algorithm}
\usepackage{algpseudocode}
\usepackage{amsmath,amssymb,amsthm}  
\usepackage{mathtools}             
\usepackage{complexity}            
\usepackage{braket}                 
\usepackage{yhmath}                 
\usepackage{algorithm}             
\usepackage{algpseudocode}         
\usepackage{graphicx}              
\usepackage{tikz, qcircuit}    
\usetikzlibrary{decorations.pathreplacing,arrows.meta}
\tikzset{
  mynode/.style={fill=orange!30, circle, minimum size=0.5mm, inner sep=5pt}
}
\usepackage{forest}
\usepackage{hyperref}             
\usepackage{cleveref}              
\usepackage{xspace}
\usepackage{enumitem}             
\usepackage{booktabs}   
\usepackage{caption}
\usepackage[style=alphabetic,backend=biber]{biblatex}  
\usepackage{xcolor}     

\theoremstyle{plain} 
\newtheorem{theorem}{Theorem}[section]
\newtheorem{lemma}[theorem]{Lemma}

\newtheorem{corollary}[theorem]{Corollary}

\theoremstyle{definition} 
\newtheorem{definition}[theorem]{Definition}

\newtheorem{fact}[theorem]{Fact}

\theoremstyle{remark}  
\newtheorem{remark}[theorem]{Remark}

\newcommand{\eps}[0]{\varepsilon}
\newcommand{\tr}[0]{\text{Tr}}
\newcommand{\pr}[0]{\text{Pr}}

\newcommand{\ketbra}[2]{|{#1}\rangle\langle{#2}|}
\newcommand{\abs}[1]{\left| #1 \right|}
\newcommand{\norm}[1]{\left\lVert #1 \right\rVert}
\newcommand{\inn}[2]{\left\langle #1,  #2 \right\rangle}

\newcommand{\pbra}[1]{\left(#1\right)}

\newcommand{\ot}{\otimes}
\newcommand{\zo}{\{0,1\}}

\newcommand{\PARITY}{\texttt{PARITY}\xspace}

\DeclareMathOperator{\Poly}{poly}

\usetikzlibrary{calc} 
\hypersetup{
    colorlinks=true,
    linkcolor=blue,
    filecolor=magenta,
    urlcolor=cyan,
    citecolor=blue
}

\title{Random Unitaries in Constant (Quantum) Time}

\author{
Ben Foxman\thanks{Yale University. \url{ben.foxman@yale.edu}} \quad\quad\quad
Natalie Parham\thanks{Columbia University. \url{natalie@cs.columbia.edu}} \quad\quad\quad
Francisca Vasconcelos\thanks{UC Berkeley. \url{francisca@berkeley.edu}} \quad\quad\quad
Henry Yuen\thanks{Columbia University. \url{hyuen@cs.columbia.edu}}
}
\date{}

\begin{document}
\maketitle
\begin{abstract}
   Random unitaries are a central object of study in quantum information, with applications to quantum computation, quantum many-body physics, and quantum cryptography. Recent work has constructed unitary designs and pseudorandom unitaries (PRUs) using $\Theta(\log \log n)$-depth unitary circuits with two-qubit gates. 

   In this work, we show that unitary designs and PRUs can be efficiently constructed in several well-studied models of \emph{constant-time} quantum computation (i.e., the time complexity on the quantum computer is independent of the system size). These models are constant-depth circuits augmented with certain nonlocal operations, such as (a) many-qubit \texttt{TOFFOLI} gates, (b) many-qubit \texttt{FANOUT} gates, or (c) mid-circuit measurements with classical feedforward control. Recent advances in quantum computing hardware suggest experimental feasibility of these models in the near future. 

   Our results demonstrate that unitary designs and PRUs can be constructed in much weaker circuit models than previously thought. 
   Furthermore, our construction of PRUs in constant-depth with many-qubit \texttt{TOFFOLI} gates shows that, under cryptographic assumptions, there is no polynomial-time learning algorithm for the circuit class $\QAC^0$.\ Finally, our results suggest a new approach towards proving that \texttt{PARITY} is not computable in $\QAC^0$, a long-standing question in quantum complexity theory. 
\end{abstract}
\section{Introduction}

 Unitary designs and pseudorandom unitaries (PRUs) are efficiently implementable ensembles of unitaries that mimic Haar-random unitaries; $t$-designs are \emph{statistically} indistinguishable from the Haar measure up to $t$'th moments, and PRUs are \emph{computationally} indistinguishable from Haar when queried by efficient quantum algorithms.
 Applications of unitary designs and pseudorandom unitaries are diverse, ranging from quantum compiling \cite{Wallman_2016}, benchmarking \cite{Dankert_2009, Magesan_2011}, quantum learning \cite{huang2023learningpredictarbitraryquantum}, quantum supremacy experiments \cite{48651}, and quantum cryptography~\cite{gunn2023commitments}. 

There is a long line of work focused on efficiently implementing unitary designs~\cite{Brand_o_2016,Haferkamp_2022,   Harrow_2023,metger2024simpleconstructionslineardepthtdesigns, laracuente2024approximateunitarykdesignsshallow,schuster2025randomunitariesextremelylow} and PRUs~\cite{schuster2025randomunitariesextremelylow, ma2024constructrandomunitaries, lu2025parallelkacswalkgenerates}. Particularly striking are the ``gluing-type'' results of LaRacuente and Leditzky \cite{laracuente2024approximateunitarykdesignsshallow} and Schuster, Haferkamp, and Huang \cite{schuster2025randomunitariesextremelylow}.\ In these works, the authors show that several small, independent copies of $t$-designs or PRUs can be ``glued’’ together, yielding random unitaries on larger numbers of qubits. This allows one to construct $t$-designs/PRUs in $O(\log n)$ depth on nearest-neighbor circuit geometries, and $O(\log \log n)$ depth in all-to-all circuits. Given that circuit depth corresponds to the time complexity of a computation, this results imply that (pseudo)random unitaries can be implemented in little time on a quantum computer. 

Schuster et al.\ ~\cite{schuster2025randomunitariesextremelylow} also show that the Gluing Lemma is optimal: unitary designs and PRUs provably require $\Theta(\log \log n)$ depth in all-to-all architectures and $\Theta(\log n)$ depth in nearest-neighbor architectures. Given this, it may appear that the story is closed on minimizing the time complexity of unitary designs and PRUs. However, this is only the limit when we consider \emph{unitary} circuits composed of single- and two-qubit gates. In this paper we show that $t$-designs and PRUs can be implemented in \emph{constant time} in several well-motivated models of quantum computation. We then describe several implications of our results, ranging from possible experimental implementations, quantum complexity theory, and limits on quantum learning.

\paragraph{Models of constant-time computation.} Constant-depth quantum circuits are a model of \emph{constant-time} quantum computations (i.e., the time complexity on the quantum computer is independent of the system size). Constant-depth circuits with only local unitary gates are quite limited in power; they cannot prepare long-range entanglement (like in the GHZ state) or compute functions that depend on all input bits. However, there has been significant interest in understanding the power of constant-depth quantum circuits augmented with complex operations such as many-qubit gates or intermediate measurements with feedforward control. The addition of nonlocal operations appears to unlock significant computational advantages for constant-depth circuits. Some models include:
\label{sec:intro-models}
\begin{enumerate}
    \item Short-time Hamiltonian evolutions. Such operations are natural in analog quantum simulators~\cite{periwal2021programmable}, and are capable of entangling many qubits simultaneously~\cite{fenner2003implementingfanoutgatehamiltonian}.
    
    \item Constant-depth circuits with mid-circuit measurements (i.e., intermediate measurements, followed by classical postprocessing of measurement outcomes to control subsequent layers of gates). This model is important for quantum error-correction~\cite{shor1995scheme}, and has been shown to be capable of generating exotic quantum states with long-range entanglement~\cite{tantivasadakarn2023hierarchy,Buhrman_2024}. Here we treat classical computation as a free operation (i.e., costing no time), which is a reasonable assumption in many cases. 

    \item Constant-depth circuits with arbitrary single-qubit gates and many-qubit \texttt{FANOUT} gates (which maps $\ket{s, x_1, \dots, x_k}$ to $\ket{s, s\oplus x_1, \dots, s \oplus x_k}$). This gives rise to the so-called $\QAC^0_f$ circuit class, which has been studied in quantum complexity theory as a quantum analogue of the classical circuit class $\AC^0$~\cite{moore1999quantumcircuitsfanoutparity}. The addition of the \texttt{FANOUT} gate makes constant-depth circuits surprisingly powerful; for example, the \texttt{PARITY} function and even the quantum Fourier transform can be computed with $\QAC^0_f$ circuits~\cite{Hoyer_2005}. 
    
    \item Circuits with constant \emph{$T$-depth}, that is, arbitrary Clifford computation interleaved with a constant number of layers of $T$ gates. Such circuits have been studied in the context of fault-tolerant quantum computation~\cite{fowler2012time}, nonlocal quantum computation~\cite{speelman2016instantaneous}, and quantum circuit lower bounds~\cite{parham2025quantumcircuitlowerbounds}. This is a reasonable model of constant-time computation in the fault-tolerance setting, because the cost of performing logical $T$ gates dominates that of Clifford operations. 
\end{enumerate}

These models might appear very different from each other, but in fact some of them are closely related. For example, constant $T$-depth circuits, when given access to a catalytic resource state, can simulate $\QAC^0_f$~\cite{kim2025catalyticzrotationsconstanttdepth}, which in turn is equivalent to constant-depth circuits with mid-circuit measurements where the classical processing consists of parity computations~\cite{Buhrman_2024}. 

The equivalence of these models greatly extends the flexibility of quantum algorithm design. For example, one could design an efficient quantum algorithm in the $\QAC^0_f$ model (which may be conceptually easier to think about), and then immediately obtain an efficient implementation in the mid-circuit measurement model, which is experimentally realizable~\cite{corcoles2021exploiting}. We take this approach to show that unitary designs and pseudorandom unitaries can be implemented in constant (quantum\footnote{The parenthetical is meant to emphasize that the \emph{quantum} part of the circuit has constant depth, and we treat the classical processing, which will be simple parity computations, as free.}) time.

\subsection{Our Contributions}
First we show that unitary designs and PRUs can be implemented by polynomial-size $\QAC^0_f$ circuits. Here, we consider two flavors of designs and PRUs, \emph{standard} and \emph{strong}. Standard means that the unitaries appear Haar-random to algorithms that can make (adaptive) queries to the forwards direction of the unitary, whereas strong means that the unitaries appear Haar-random even when the algorithm can make queries to the inverse unitary as well~\cite{ma2024constructrandomunitaries}. 

This theorem circumvents the lower bounds of~\cite{schuster2025randomunitariesextremelylow} which only apply to circuits with single- and two-qubit unitary gates.

\begin{theorem}[Random Unitaries in $\QAC^0_f$]
\label{thm:intro-strong-random-unitaries-in-qac0f}
Let $n$ denote the number of input qubits. The following can be implemented in $\QAC^0_f$, the class of constant-depth circuits with arbitrary single-qubit gates and many-qubit \texttt{FANOUT} gates:
\begin{itemize}
    \item Standard $1/\Poly(n)$-approximate $t$-designs can be implemented with $n \Poly\log n \cdot \Poly(t)$-size circuits; and
    \item Strong $\exp(-\Omega(n))$-approximate $t$-designs can be implemented with $\Poly(n,t)$-size circuits.
\end{itemize}
Furthermore, assuming subexponential post-quantum security of the Learning With Errors (LWE) problem, the following can also be implemented in $\QAC^0_f$:
\begin{itemize}
    \item Standard PRUs can be implemented with $n \Poly \log n$-size circuits; and
    \item Strong PRUs can be implemented with $\Poly(n)$-size circuits.
\end{itemize}
\end{theorem}

The precise definitions of approximate unitary designs and PRUs can be found in \Cref{sec:random-unitaries-defs}, and we prove \Cref{thm:intro-strong-random-unitaries-in-qac0f} in \Cref{sec:random-unitaries-in-qac0}. We stress that while the designs are implemented by constant-depth $\QAC^0_f$ circuits, the pseudorandomness holds against \emph{all} algorithms that make at most $t$ queries; in particular the algorithms can perform arbitrarily deep computation in between queries. 

The result for $t$-designs and PRUs are stated and proved in \Cref{thm:fdesign} and \Cref{thm:fpru} respectively. We note that a similar approach was used by Chia, Liang, and Song \cite{chia2024quantumstatelearningimplies} to obtain random unitaries in $\QAC^0_f$, under slightly different assumptions. 

The close connection between $\QAC^0_f$ and circuits with mid-circuit measurements yields the following corollary: unitary designs and PRUs can also be constructed in constant-depth 2D-local circuit geometries using nearest-neighbor gates, intermediate measurements, and feedforward operations. 

\begin{corollary}[Random unitaries in the intermediate measurement model]
    \label{cor:intro-msmt}
    Approximate $t$-designs and PRUs can be implemented by constant-depth quantum circuits with nearest-neighbor two-qubit gates on a 2D lattice, intermediate measurements, and classical feedforward consisting of computing parities of measurement outcomes. The circuits for $t$-designs have size $n \Poly\log n \cdot \Poly(t)$, and the circuits for PRUs have size $n \Poly\log(n)$. In the same model, strong $t$-designs and PRUs can be implemented by circuits with size $\Poly(n,t)$ and $\Poly(n)$ respectively.
\end{corollary}
\noindent We prove \Cref{cor:intro-msmt} in \Cref{sec:msmt}. 

\begin{remark}
    It is important to note that, even though the mid-circuit measurements yield random outcomes, \Cref{cor:intro-msmt} gives a \emph{deterministic} implementation of every unitary in the $t$-design/PRU ensemble. The measurement outcomes are only used for simple Pauli corrections, but don't change the overall unitary that is being performed. 
\end{remark}

\Cref{thm:intro-strong-random-unitaries-in-qac0f} and \Cref{cor:intro-msmt} provides another illustration of the power of models of constant-time quantum computation. There is a long line of work demonstrating the various capabilities of $\QAC^0_f$ circuits and circuits with mid-circuit measurement, from the quantum Fourier transform~\cite{Hoyer_2005} to the preparation of novel states from condensed matter physics~\cite{tantivasadakarn2023hierarchy,Buhrman_2024}. We can now add $t$-designs and PRUs to this list.

Furthermore, our constructions suggest the possibility of highly efficient experimental implementations of random unitaries. Random unitaries have been extremely useful in a variety of experimental settings, such as benchmarking~\cite{Dankert_2009, Magesan_2011}, demonstrations of quantum supremacy~\cite{48651}, and explorations of quantum chaos~\cite{belyansky2020minimal} and quantum thermalization~\cite{green2022experimental}. A number of experimental platforms are starting to support mid-circuit measurements and feedforward control~\cite{corcoles2021exploiting,Iqbal_2024}, which may bring the protocol of \Cref{cor:intro-msmt} within the realm of near-term implementation. 

\paragraph{Random unitaries in $\QAC^0$.} In~\cite{moore1999quantumcircuitsfanoutparity}, Moore also introduced a variant of $\QAC^0_f$ called $\QAC^0$ -- note the lack of the $f$ subscript. This model consists of constant-depth circuits with arbitrary single-qubit gates and (instead of the \texttt{FANOUT} gate) many-qubit \texttt{TOFFOLI} gates, which maps $\ket{x_1,\ldots,x_k,s}$ to $\ket{x_1,\ldots,x_k, s \oplus \texttt{AND}(x_1, \ldots,x_k)}$. Both $\QAC^0$ and $\QAC^0_f$ were defined as natural analogues of the classical circuit class $\AC^0$, an extremely well-studied circuit model in theoretical computer science~\cite{FSS84}.  

$\QAC^0$ can be efficiently simulated by $\QAC^0_f$ circuits\footnote{Moore originally defined $\QAC^0_f$ to include \emph{both} \texttt{FANOUT} and \texttt{TOFFOLI} gates, which meant that $\QAC^0$ is contained in $\QAC^0_f$ by definition. However, Takahashi and Tani later showed that \texttt{TOFFOLI} gates can be simulated using \texttt{FANOUT} gates~\cite{takahashi2012collapsehierarchyconstantdepthexact}.}, but it is believed that $\QAC^0$ is much weaker than $\QAC^0_f$. It is conjectured, for example, that the $\texttt{PARITY}$ function cannot be computed by polynomial-size $\QAC^0$ circuits, implying that $\QAC^0 \neq \QAC^0_f$. Despite $\QAC^0$ appearing to be such a weak circuit class, we show that such circuits can nevertheless efficiently implement $t$-designs for $t = O(1)$ and PRUs with inverse-polynomial approximation error. 

\begin{theorem}[Random unitaries in $\QAC^0$]
\label{thm:intro-designs-in-qac0}
Let $n$ denote the number of input qubits. Let $\delta > 0$ and let $k \in \mathbb{N}$ be constants. The following can be implemented in $\QAC^0$, the class of constant-depth circuits with arbitrary single-qubit gates and many-qubit \texttt{TOFFOLI} gates:
\begin{itemize}
    \item Standard $n^{-k}$-approximate $t$-designs can be implemented with depth $O((\log k)/\delta)$, size $O(n^{1 + \delta t})$ circuits. 
\end{itemize}
Furthermore, assuming subexponential post-quantum security of the Learning With Errors (LWE) problem, the following can also be implemented in $\QAC^0$:
\begin{itemize}
    \item Standard PRUs can be implemented up to $n^{-k}$ error, by depth $O((\log k)/\delta)$, size $O(n^{1 + \delta})$  circuits. 
\end{itemize}
\end{theorem}

\begin{remark}
We note that the PRUs in \Cref{thm:intro-designs-in-qac0} can only be implemented by $\QAC^0$ circuits up to inverse polynomial error (whereas the $\QAC^0_f$ circuits in \Cref{thm:intro-strong-random-unitaries-in-qac0f} can implement PRUs with \emph{exponentially} small error). This makes the PRU construction less relevant for cryptographic applications, perhaps, but we believe it still has the following use cases:
\begin{enumerate}
    \item The $\QAC^0$ $t$-design constructions only look random up to $t=O(1)$ moments; however the PRU constructions look random (to computationally-bounded adversaries) up to $O(n^k)$ queries, and
    \item The PRU constructions imply lower bounds on the complexity of \emph{learning} $\QAC^0$ circuits, which we elaborate on shortly.
\end{enumerate}
\end{remark}

We find \Cref{thm:intro-designs-in-qac0} interesting for several reasons.\ The first is that there seems to be no classical analogue. Pseudorandom functions cannot be efficiently implemented by classical $\AC^0$ circuits~\cite{10.1145/174130.174138}, which are constant-depth circuits with unbounded fan-in \texttt{AND} and \texttt{NOT} gates (i.e., the classical analogue of $\QAC^0$). Thus \emph{quantum} pseudorandomness appears easier to achieve in a weak quantum circuit model than \emph{classical} pseudorandomness in an analogous classical circuit model. We discuss in more detail in \Cref{sec:intro-discussion}.

A second reason is that, to the best of our knowledge, \Cref{thm:intro-designs-in-qac0} demonstrates the first example of an interesting task implementable in $\QAC^0$. Moreover, for constant $t$-designs and PRUs, the number of ancillae is $n^{1 + 1/O(d)}$, where $d$ is the depth of the $\QAC^0$ circuit. By a result of Anshu et al.\ \cite{anshu2024computationalpowerqac0barely}, the \texttt{PARITY} function cannot be computed by $\QAC^0$ circuits with $n^{1 + 1/\exp(d)}$ ancillae. Therefore, there number of ancillae in our constructions are quite close to the current best-known \texttt{PARITY} lower bounds. 

A third reason is due to a connection with quantum complexity theory. Given \Cref{thm:intro-designs-in-qac0}, one may naturally wonder about whether $\QAC^0$ can implement strong designs or strong PRUs? Fascinatingly, this connects to a long-standing open question in quantum circuit complexity: can $\QAC^0$ efficiently simulate $\QAC^0_f$, or equivalently, is the $\texttt{PARITY}$ function computable in $ \QAC^0$~\cite{moore1999quantumcircuitsfanoutparity}? If $\QAC^0 = \QAC^0_f$, then by \Cref{thm:intro-strong-random-unitaries-in-qac0f} certainly strong designs/PRUs can be implemented in $\QAC^0$. Conversely, if one can show that strong designs or PRUs \emph{cannot} be implemented by $\QAC^0$ circuits, then this proves $\texttt{PARITY} \notin \QAC^0$. 

\paragraph{Hardness of learning $\QAC^0$.}
A central problem in quantum learning theory is reconstructing the description of a quantum circuit given only black-box query access. As alluded to above, \Cref{thm:intro-designs-in-qac0} has direct implications for learning $\QAC^0$ circuits—a topic explored in several prior works \cite{Nadimpalli_2024, bao2025learningjuntadistributionsquantum, vasconcelos2025learningshallowquantumcircuits}. In particular, assuming the hardness of the Learning with Errors (LWE) problem, the following theorem leverages our $\QAC^0$ PRUs to establish super-polynomial time lower bounds for \emph{average-case} learning of polynomial-size $\QAC^0$ circuits.

\begin{theorem}[$\QAC^0$ Learning Lower-Bound]
\label{thm:intro-learning-lower-bound}
    Assuming subexponential post-quantum security of LWE, $n^{\omega(1)}$ time is necessary to learn $\QAC^0$ unitaries with $O(n^\delta)$ ancillae to average-case distance\footnote{Here average-case distance refers to the error in average gate-fidelity according to Haar random inputs, as described in depth in \cite{vasconcelos2025learningshallowquantumcircuits}.}, for any constant $\delta > 0$.
\end{theorem} 
\noindent This sublinear-ancillae regime is very close to the logarithmic-ancillae regime for which the \cite{vasconcelos2025learningshallowquantumcircuits} algorithm gives a nearly matching learning upper-bound. By assuming the quantum LMN low-support conjecture \cite[Conjecture 1]{vasconcelos2025learningshallowquantumcircuits}, their learning guarantees extend to the polynomial-ancillae regime, thereby proving near-optimality of the \cite{vasconcelos2025learningshallowquantumcircuits} $O(n^{\Poly\log(n)})$ runtime. Finally, we emphasize that this is the first \emph{average-case} learning lower-bound for $\QAC^0$, with previous work by \cite{10.1145/3618260.3649722,vasconcelos2025learningshallowquantumcircuits} establishing an exponential sample worst-case learning lower-bound.

We refer the reader to \Cref{sec:intro-discussion} for additional context. 

\paragraph{Construction overview.} We provide a brief overview of our approach. To construct (strong) unitary designs and (strong) PRUs in $\QAC^0_f$ (\Cref{thm:intro-strong-random-unitaries-in-qac0f}), we parallelize a variant of the so-called $PFC$ construction of random unitaries, first introduced by~\cite{metger2024simpleconstructionslineardepthtdesigns} and further studied by~\cite{ma2024constructrandomunitaries, cui2025unitarydesignsnearlyoptimal}. Here, $P$, $F$ and $C$ represent a random permutation matrix, a random diagonal matrix, and a random Clifford operator respectively. We show that the ensemble of $PFC$ matrices analyzed by~\cite{cui2025unitarydesignsnearlyoptimal} can all be implemented in $\QAC^0_f$; this exploits the surprising power of the \texttt{FANOUT} gate~\cite{Hoyer_2005,grier2024quantumthresholdpowerful}. 

To get the designs in $\QAC^0$ (without the \texttt{FANOUT} gate), we modify Morris and Grier's \emph{exponential}-size $\QAC^0$ implementation of the \texttt{FANOUT} gate~\cite{grier2024quantumthresholdpowerful} to convert the $\QAC^0_f$ constructions above into exponential-size $\QAC^0$ circuits. We then scale these down to logarithmically-many input qubits (yielding polynomial-size $\QAC^0$ circuits), and then apply the ``Gluing Lemma’’ of~\cite{schuster2025randomunitariesextremelylow} to recover a unitary design on $n$ qubits. 

Random unitaries constructed via the Gluing Lemma are not secure when inverse queries are allowed, so \Cref{thm:intro-designs-in-qac0} only yields the standard notion of designs (rather than strong designs). 
 \subsection{Discussion}
 \label{sec:intro-discussion}

We now elaborate on the implications and consequences of our constructions of random unitaries in $\QAC^0$ and $\QAC^0_f$. 
 
\paragraph{Classical versus Quantum Pseudorandomness.} The quantum circuit models $\QAC^0$ and $\QAC^0_f$ were first introduced by Moore to investigate quantum analogues of the classical circuit class $\AC^0$, which are constant-depth classical circuits with unbounded fan-in \texttt{AND} and \texttt{NOT} gates. The study of $\AC^0$ has been extremely fruitful in theoretical computer science, yielding celebrated lower bounds such as $\texttt{PARITY} \notin \AC^0$~\cite{FSS84,Has86}, powerful techniques such as random restrictions and the polynomial method~\cite{Has86,Razborov1987LowerBO}, and subsequent developments in learning theory and complexity theory that continue to this day~\cite{10.1145/174130.174138,braverman2008polylogarithmic,raz2022oracle}. 

The existence of random unitaries in $\QAC^0$ (\Cref{thm:intro-designs-in-qac0}) is particularly surprising given that an analogous result provably does not hold for $\AC^0$. In particular, Linial, Mansour, and Nisan showed that a simple, 2-query algorithm exists for distinguishing (quasipolynomial size) $\AC^0$ circuits from random functions. In particular, \cite{10.1145/174130.174138} showed that every $\AC^0$ function is approximately a low-degree polynomial, and as a consequence, it has low average sensitivity—for a random input $x$, flipping a random bit of to obtain $x^\prime$ is unlikely to affect the value of the function. Checking if $f(x) = f(x^\prime)$ suffices to distinguish $f$ from a random function, where output bits are independently uniform random. 
\par 
Our results imply that no analogous distinguisher should exist in the quantum setting, when only forward queries to $U$ are allowed. This is because the polynomial-size designs from \Cref{thm:intro-designs-in-qac0} will fool any constant-query distinguishers, and the PRUs will fool polynomial-time distinguishers. 

Perhaps one interpretation is this: quantum \texttt{AND} (i.e., \texttt{TOFFOLI}) gates appear more effective at generating \textit{quantum randomness} than classical \texttt{AND} gates are at generating \textit{classical randomness}.

 \paragraph{Towards Quantum Circuit Lower Bounds.} 
Is the \texttt{PARITY} function computable in $\QAC^0$? In his paper introducing $\QAC^0$ and $\QAC^0_f$~\cite{moore1999quantumcircuitsfanoutparity}, Moore showed that, interestingly, computing the $\texttt{PARITY}$ function is \emph{equivalent} to the \texttt{FANOUT} gate (and thus $\texttt{PARITY} \in \QAC^0_f$), and thus the question about \texttt{PARITY} is equivalent to whether the two models $\QAC^0$ and $\QAC^0_f$ are equal in power. 

The question remains stubbornly open 25 years after Moore initially posed the question. The classical techniques for proving $\texttt{PARITY} \notin \AC^0$ do not seem to carry over to the quantum setting. Lower bounds have been given in some restricted settings, such as when the circuit has depth at most 2 \cite{pade2020depth,rosenthal2020boundsqac0complexityapproximating,fenner2025tightboundsdepth2qaccircuits}, or the number of ancillae is limited \cite{fang2003quantumlowerboundsfanout,Nadimpalli_2024,anshu2024computationalpowerqac0barely}. However, even depth $3$ lower bounds remain open. 

In this work, we show that $\texttt{PARITY} \notin \QAC^0$ is implied by the \emph{non-existence} of strong random unitaries in $\QAC^0$. While a formally a harder task, the connection with random unitaries may unlock new approaches to proving $\QAC^0$ lower bounds. We elaborate on these approaches in the next section. 

To show this connection, we leverage Rosenthal's $\QAC^0$ reduction from implementing \texttt{PARITY} to synthesizing a special state known as a \textit{nekomata} state. If the parity \textit{decision problem} was computable in $\QAC^0$, such states could be constructed, which would give both \texttt{PARITY} and strong random unitaries in $\QAC^0$.  

\paragraph{Implications of Inverse Access.} 
Our constructions of random unitaries in $\QAC^0$ are secure against adversaries with only forward access to $U$. When \texttt{FANOUT} gates are allowed, we obtain random unitaries which are secure even with inverse access. Therefore, if one could show that $\QAC^0$ circuits differ from Haar-random for any property which is easy to estimate with inverse access, then the \texttt{FANOUT} gate, and therefore \texttt{PARITY} gate, cannot be computed in $\QAC^0$.\ As described previously, uncomputation techniques and a $\QAC^0$ reduction due to Rosenthal (\Cref{lem:neko}) would then allow one to obtain lower bounds for \textit{decision problems} in $\QAC^0$.
\begin{table}[ht]
    \centering
    \begin{tabular}{l|p{0.65\linewidth}}
        \textbf{Access Model} & \textbf{Capability} \\
        \hline
        $U$ only & Collision probability, purity of subsystems, average fidelity across inputs, overlap with Pauli operators~\cite{huang2023learningpredictarbitraryquantum} \\
        \hline
        $U$ and $U^\dagger$ & Determining lightcones~\cite{schuster2025randomunitariesextremelylow}, entanglement entropy~\cite{gur2021sublinearquantumalgorithmsestimating,chen2025localtestunitarilyinvariant}, OTOCs~\cite{cotler2022informationtheoretichardnessoutoftimeordercorrelators}, displacement amplitudes~\cite{King_2024} \\
    \end{tabular}
    \caption{Estimating properties with forward queries to just $U$, or queries to both $U$ and $U^\dagger$. In the latter case, superpolynomial queries are required for each task when only forward queries are allowed.}
    \label{tab:access-models}
\end{table}

The existence of forward-secure (pseudo)random unitaries in quasipolynomial size $\QAC^0$ implies that all statistics in the first row of \Cref{tab:access-models} must approximately agree with the Haar-random values.\ The same is true of random unitaries in $\QAC^0_f$ with respect to \textit{all} properties in the table. Therefore, one cannot hope to prove lower bounds against $\QAC^0$ or $\QAC^0_f$ by showing that these properties significantly differ from the expected  Haar-random values. This is reminiscent of the \textit{natural proofs barrier} introduced by Razborov and Rudich~\cite{MR1473047}, which has been considered in other contexts in quantum information~\cite{dowling2006geometryquantumcomputation, chen2025finegrainedcomplexityquantumnatural}.\ Defining quantum natural proofs in the context of $\QAC^0$ may require addressing the subtle distinction between forward and inverse queries—we leave this for future work.

\paragraph{Learning Lower-Bounds.} Classically, \cite{10.1145/174130.174138} leveraged low-degree Fourier concentration to derive the first efficient (quasipolynomial) sample and time complexity algorithm for learning $\AC^0$ functions. Assuming the classical hardness of factoring, \cite{10.1145/167088.167197} proved that quasipolynomial sample complexity of learning $\AC^0$ is optimal. More recent follow-up work of \cite{doi:10.1137/20M1344202} also proved a quasipolynomial time lower-bound for \emph{quantum} learning of $\AC^0$, assuming sub-exponential post-quantum security of Ring Learning With Errors (RLWE). 

On the quantum side, there has been substantial recent progress on learning $\QAC^0$—\cite{Nadimpalli_2024} gave a quasipolynomial time algorithm for learning single-output channel Choi representations generated by $\QAC^0$ circuits with logarithmic ancillae, which was improved by Vasconcelos and Huang \cite{vasconcelos2025learningshallowquantumcircuits} to obtain a quasipolynomial \textit{time} algorithm for \textit{average-case} learning of full $\QAC^0$ (with logarithmic ancillae) unitaries. Furthermore, assuming a quantum analog of the LMN Theorem, they showed that the algorithm can be extended to any polynomial-size $\QAC^0$ circuit. 

Assuming subexponential-security of LWE, we show that quasipolynomial time complexity is near optimal for average-case learning of polynomial-size $\QAC^0$ circuits. We present this as a natural corollary of our PRU constructions, although the same result can also be obtained more directly using a single, scaled-down quantum secure PRF. The latter approach can be viewed as a ``fully quantum'' analog of the technique used in \cite{doi:10.1137/20M1344202}.

\subsection{Conclusion and Further Directions}
\label{sec:intro-directions}
Overall, our results highlight the surprising power of constant-depth circuits with many-qubit gates, suggesting the shallow quantum circuits may be even more powerful than previously thought. Since many quantum technologies are expected to support many-qubit gates \cite{song2024realizationconstantdepthfanoutrealtime, Bluvstein_2022, yu2022multiqubittoffoligatesoptimal}, our results suggest that a staple of quantum information processing, (pseudo)random unitaries, may be realizable in constant time on future quantum hardware. 
Moreover, our constructions reveal new connections between quantum (pseudo)randomness and quantum circuit complexity. 
 
We believe that our results suggest some interesting directions for future investigation, which we list below. 

\paragraph{Improving the Constructions.} Current constructions of the \texttt{FANOUT} gate in $\QAC^0$~\cite{rosenthal2020boundsqac0complexityapproximating, grier2024quantumthresholdpowerful} require size that scales inverse polynomially in implementation error. Improving the $\eps$ dependence could drastically improve the parameters for random unitaries implementable in $\QAC^0$ (see \Cref{sec:random-unitaries-in-qac0}). Moreover, could a fine-grained ancilla-depth tradeoff be established? This would be relevant for experimental considerations.

\paragraph{Inverse Security and $\QAC^0$.} As mentioned in \Cref{sec:intro-discussion}, proving that strong $t$-designs are not in $\QAC^0$ is a path to proving $\PARITY \notin \QAC^0$. In contrast with recent approaches~\cite{Nadimpalli_2024,anshu2024computationalpowerqac0barely} which aim to generalize classical techniques to the quantum setting, this route would not require such an analogy. Could properties in \Cref{tab:access-models} be helpful for ruling out strong $t$-designs in $\QAC^0$, and by extension proving \texttt{PARITY} $\notin \QAC^0$? 

\paragraph{Weakening Assumptions for Random Unitaries.} 
Even though $\AC^0$ does not contain PRFs, there are candidate constructions of \textit{weak} PRFs in $\AC^0$ \cite{BlumFurstKearnsLipton:1994}, which are secure against adversaries who only receive random $(x, f(x))$ samples as opposed to arbitrary query access. Could PRUs be built from post-quantum \textit{weak} PRFs? Currently, it is only known how to build PRUs from standard quantum-secure PRFs, which are secure against adversaries making adaptive queries. If so, then the Gluing Lemma may no longer be necessary, and \textit{strong} random unitaries could be obtained in $\QAC^0$. Conversely, this could serve as a path towards ruling out the existence of quantum-secure weak PRFs in $\AC^0$. 

\paragraph{Stronger Learning Lower-Bounds.} 
Leveraging our efficient $\QAC^0$ $t$-design implementations, it could be interesting to establish information theoretic learning lower-bounds, using arguments like those of \cite{chen2025information}, to potentially rule out efficient low-degree learning of $\QAC^0$. Furthermore, as shown in this work, PRU implementations can yield powerful \emph{average-case} computational learning lower-bounds for $\QAC^0$. More ancilla-efficient $\QAC^0$ PRU/PRF implementations could improve the learning lower-bound on two fronts: 1) extending it to $\QAC^0$ with only logarithmic ancillae, where there is no reliance on the quantum LMN conjecture, and 2) closing the current gap with the quasi-polynomial upper-bound. Achieving both of these would establish optimality of the \cite{vasconcelos2025learningshallowquantumcircuits} algorithm.

\paragraph{The ``Quantum Natural Proofs Barrier''.} \Cref{tab:access-models} lists many properties which, by our results, must agree with Haar-random values for $\QAC^0$ and $\QAC^0_f$. This rules out several approaches for proving lower bounds against these classes. Are there other approaches for proving quantum circuit lower bounds which circumvent these barriers? On the other hand, can we find other barriers for proving lower bounds against low-depth quantum circuits?  

\paragraph{How Easy is Generating Quantum Pseudorandomness?}
A broader takeaway from our paper is the idea that random unitaries can be implemented in a much wider range of constant-time models of quantum computation than previously thought. Could current-day programmable analog quantum simulators, which can perform globally entangling Hamiltonian evolutions~\cite{periwal2021programmable}, be sufficient to generate quantum pseudorandomness in constant time? Recent developments in using optical cavities with neutral atoms may offer further avenues to perform strongly coupled interactions between large numbers of qubits~\cite{li2022collective,Jandura_2024} -- could such interactions be used in place of the \texttt{TOFFOLI} or \texttt{FANOUT} gates to efficiently implement designs and PRUs in neutral atom arrays? Theoretically, can we characterize the weakest possible entangling operation that suffices to generate quantum pseudorandomness in constant time? These are interesting questions that we leave for future work.

\section{Background}
In this paper, $\exp$ and $\log$ denote base-$2$ exponent and logarithm. $\Poly(n) = \bigcup_{c \in \mathbb{N}} O(n^c)$. A function $f(n)$ is \textit{negligible} if $f(n) \in o(1/p(n))$ for any polynomial $p$. A function is \textit{quasipolynomial} if it is of the form $n^{\Poly \log n}$. 
\subsection{Quantum Computation}

We assume basic familiarity with the formalism of quantum computing, including states, unitaries, and measurements. In this section, we briefly recall concepts that will be important for this work. First, the \textit{Pauli group} on $n$ qubit unitaries consists of all $n$-fold tensor products of the single-qubit matrices $\{I, X, Y, Z\}$, multiplied by a global phase in $\{\pm 1, \pm i\}$.\ A \textit{Pauli string} refers to an element of this group with phase $1$.\ The set of Pauli strings is an orthonormal basis for the vector space of $2^n \times 2^n$ matrices, under the Hilbert-Schmidt inner product $\inn{A}{B} = \frac{1}{2^n}\tr(A^\dagger B)$.\ A \textit{Clifford operator} is a unitary that normalizes the Pauli group, i.e.\ $CPC^\dagger$ is in the Pauli group if $P$ is.\ The set of $n$-qubit Clifford operators forms a group generated by the Hadamard ($H$), phase ($S$), and \texttt{CNOT} gates. \par The \textit{trace distance} between two quantum states $\rho$ and $\sigma$ is defined as
\[
\|\rho - \sigma\|_{td} := \frac{1}{2} \|\rho - \sigma\|_{1},
\]
where $\|A\|_1 = \operatorname{Tr}[\sqrt{A^\dagger A}]$ denotes the \text{trace norm}.\ Operationally, the trace distance is equivalent to  the maximum distinguishing advantage between $\rho$ and $\sigma$ over all possible quantum measurements, i.e. $\|\rho - \sigma\|_{td} \le \max_{M_i} \norm{p^\rho, p^{\sigma}}_1$ where $p^\rho, p^{\sigma}$ are classical probability distributions on measurement outcomes in the basis specified by $\{M_i\}_i$.\ The \textit{Fuchs-van-de Graaf inequality} relates the trace distance $\norm{\rho - \sigma}_{td}$ and the \textit{fidelity} $F(\rho, \sigma) := (\tr\sqrt{\sqrt{\rho}\sigma \sqrt{\rho}})^2$ via the bound
\[
\norm{\rho - \sigma}_{td} \le \sqrt{1 - F(\rho, \sigma)}.
\]
If $\rho = \ket{\psi_1}\bra{\psi_1}$ and $\sigma = \ket{\psi_2}\bra{\psi_2}$, then $F(\rho, \sigma) = \abs{\braket{\psi_1|\psi_2}}^2$. \par Given a bipartite state $\rho_{AB}$ on $\mathcal{H}_A \otimes \mathcal{H}_B$, the \textit{partial trace} over subsystem $B$ is the map $\tr_B : \mathcal{L}(\mathcal{H}_A \otimes \mathcal{H}_B) \to \mathcal{L}(\mathcal{H}_A)$ defined by
\[
\tr_B(\rho_{AB}) = \sum_j (\mathbb{I}_A \otimes \bra{j}) \rho_{AB} (\mathbb{I}_A \otimes \ket{j}),
\]
where $\{\ket{j}\}$ is an orthonormal basis for $\mathcal{H}_B$. The result is the \textit{reduced density matrix} on $\mathcal{H}_A$, describing the marginal state when subsystem $B$ is ignored.
\par 
Finally, the \textit{Haar measure}, denoted $\mu_{2^n}$, is the unique left and right unitarily invariant probability measure over the unitary group $\mathcal{U}(2^n)$. 
\subsection{Quantum Circuits}
We assume familiarity with the circuit model of quantum computation, see \cite{Nielsen_Chuang_2010} for an introduction.\ For quantum circuits, we use $n$ to denote the number of input qubits, and $a$ to denote the number of ancilla qubits.\ Quantum circuits, like their classical counterparts, can be used to compute (boolean) functions by measuring a designated output qubit in the standard basis. They can also implement more general functions—namely, unitary transformations mapping $n$ qubits to $n$ qubits. In this paper, we consider both types of computation:
\begin{definition}
    A quantum circuit $C$ computes a function $f: \{0, 1\}^n \mapsto \{0, 1\}^m$ with average-case error $\eps$ if 
    $$\pr_{x \in \{0, 1\}^n}[C(x) = f(x)] \ge 1 -\eps$$
    Here, $C(x)$ is a random variable denoting the measurement of $m$ designated output qubits in the standard basis. The randomness is taken over both the input $\ket{x}$ and the randomness in the measurement. 
\end{definition}
\begin{definition}
    \label{def:approximplement}
    An $n + a$ qubit circuit $C$ \textit{$\eps$-implements} an $n$-qubit unitary $U$ if
    \begin{equation*}
        \abs{\bra{\psi, 0^a}C^\dagger(U\ket{\psi} \otimes \ket{0^a})}^2 \ge 1 - \eps
    \end{equation*}
     for every $n$-qubit state $\ket{\psi}$.
\end{definition}
When $\eps = 0$, we simply say $C$ implements $U$.\ Approximate implementations have a composition property: we defer the proof to \Cref{app:compose}.
\begin{lemma}
    \label{lem:compose}
    Assume that $C$ $\eps_0$-implements $U$ and $D$ $\eps_1$-implements $V$. Then, $CD$ $(\eps_0 + \eps_1)$-implements $UV$. 
\end{lemma}

In this work, we study quantum analogs of the classical circuit class $\AC^0$, which are constant-depth circuits composed of unbounded fan-in \texttt{AND}/\texttt{OR} gates and \texttt{NOT} gates. Because these circuits are classical, bits can be freely copied and reused throughout the computation, and a single bit may serve as input to multiple gates in the next layer. In other words, the \texttt{FANOUT} operation
\[
(x, 0, \dots, 0) \mapsto (x, x, \dots, x)
\]
is considered free. Since quantum circuits consist of unitary gates-- which map $k$-qubits states to $k$-qubit states in a reversible way-- the output of a gate can be used in at most one subsequent gate. Moreover, due to the no-cloning theorem, quantum information cannot be freely copied. As a result, implementing a ``quantum \texttt{FANOUT}''\footnote{Quantum information can only be copied in a fixed basis due to the no-cloning theorem.} requires explicitly including a \texttt{FANOUT} gate in the gate set.\ This distinction is captured by two different generalizations of $\AC^0$ to the quantum setting: $\QAC^0$ and $\QAC^0_f$.
\begin{definition}[$\QAC^0$]
    A $\QAC^0$ circuit is a constant-depth circuit consisting of arbitrary single-qubit gates and $\mathtt{TOFFOLI}$ gates, where
    \[
    \mathtt{TOFFOLI}\ket{x_1, \dots, x_k, t} = \ket{x_1, \dots, x_k, t \oplus \mathtt{AND}(x_1, \dots, x_k)}.
    \]
    Importantly, we allow the number of ``control'' qubits $k$ to be arbitrary.
\end{definition}
Unless otherwise specified, $\QAC^0$ circuits have polynomial size, which is sometimes taken to be the definition of the class\footnote{Formally, we consider $\QAC$ circuit families $\{C_n\}$ indexed by input size $n$, each of which has $\Poly(n)$ size and $O(1)$ depth.}.
 
\begin{definition}[$\QAC^0_f$]
    A $\QAC^0_f$ circuit is a constant depth circuit consisting of arbitrary single-qubit gates, $\mathtt{TOFFOLI}$ gates, and $\mathtt{FANOUT}$ gates, where
    \[
    \mathtt{FANOUT}\ket{s, x_1, \dots, x_k} = \ket{s, s \oplus x_1, \dots, s \oplus x_k}.
    \]
\end{definition}
Like $\QAC^0$, $\QAC^0_f$ circuits have polynomial size unless specified otherwise.  

Since \texttt{FANOUT} allows classical bits to be copied, $\QAC^0_f$ can simulate polynomial size, constant-depth $\AC$ circuits, i.e. $\AC^0$. However, unlike $\AC^0$, $\QAC^0_f$ circuits can \textit{exactly} compute the $\texttt{PARITY}$ function $$\texttt{PARITY}\ket{x_1, \dots, x_k, t} \mapsto \ket{x_1, \dots, x_k, t \oplus \bigoplus_{i=1}^k x_i}$$ due to the following identity: 
\begin{fact}
    \label{fact:par_fanout}
    $H^{\otimes n}\mathtt{PARITY}_{A_1 \dots, A_{n-1}, B} H^{\otimes n} = \texttt{FANOUT}_{B, A_{n-1}, \dots, A_{1}}$.
\end{fact}
Building off of this identity, a long line of work \cite{cleve2000fastparallelcircuitsquantum, takahashi2012collapsehierarchyconstantdepthexact,  Buhrman_2024, grier2024quantumthresholdpowerful, Smith_2024} has shown that many other important operations can be simulated exactly by constant-depth circuits with $\mathtt{FANOUT}$ gates. Of particular interest to us is the \texttt{THRESHOLD}$_t$ gate, which flips a target bit if and only if the Hamming weight of the input is at least $t$.\ Amazingly, constant-depth \texttt{THRESHOLD} circuits (with arbitrary values of $t$) can be simulated exactly by $\QAC^0_f$ circuits \cite{takahashi2012collapsehierarchyconstantdepthexact}. We can state this formally by introducing the complexity class $\TC^0$: 
\begin{definition}[$\TC^0$]
    $\TC^0$ is the class of polynomial size, constant depth classical circuits with unbounded \texttt{AND}, \texttt{OR}, and $\texttt{THRESHOLD}_t$ gates for arbitrary $t$. 
\end{definition}
\begin{theorem}[\cite{Buhrman_2024, takahashi2012collapsehierarchyconstantdepthexact}]
    \label{thm:tc0_power}
     For any $f \in \TC^0$, the unitary operator $$U_f: \ket{x} \mapsto (-1)^{f(x)}\ket{x}$$ can be implemented by a $\QAC^0_f$ circuit. 
\end{theorem}
\Cref{thm:tc0_power} is particularly interesting in the context of random unitaries, since constant-depth circuits with threshold gates can implement candidate quantum-secure one way functions. 
\subsection{Random Unitaries}
\label{sec:random-unitaries-defs}

Haar-random unitary operators play a central role in quantum information \cite{48651, schuster2025randomunitariesextremelylow, onorati2019random, saini2022quantum}. However, even approximating a Haar-random unitary typically requires exponentially many many-qubit gates, making such unitaries impractical for most applications.\ Since real-world use cases require efficient circuits, we must relax the requirement of full randomness to something that is ``good enough'' for the task at hand. \par 
To obtain efficiently implementable approximations of random unitaries, two main relaxations are commonly considered:\ unitary $t$-designs~\cite{Dankert_2009} and pseudorandom unitaries~\cite{2003Sci...302.2098E}.\ The former restricts the number of times the algorithm queries $U$, while the latter restricts the algorithm's runtime. We will also use the notion of a \textit{strong} unitary $t$-design or \textit{strong} PRU, which allows the algorithm to query both $U$ and its inverse $U^\dagger$, while maintaining the same security guarantees. We now give more formal definitions: 
\begin{definition}[(Strong) Approximate Unitary $t$-designs]
\normalfont
    A family of unitaries $\mathcal{U} = \{U_{n, k}\}_{n, k}$ is a \emph{(strong) $\eps$-approximate unitary $t$-design }if for all $n$ and for any quantum algorithm $Q$ making at most $t$ queries to $U$ (and its inverse $U^\dagger)$,
    \begin{equation*}
        \abs{\underset{k \sim \mathcal{K}}{\text{Pr}}[Q^{U_{n, k},\; (U^\dagger_{n, k})}(1^n) = 1] - \underset{U \sim \mu_N}{\text{Pr}}[Q^{U, \; (U^\dagger)}(1^n) = 1]} \leq \eps
    \end{equation*}
\end{definition}
In the literature, this definition coincides with the recently defined notion of \emph{measurable}-error approximate $t$-design \cite{cui2025unitarydesignsnearlyoptimal}, to distinguish it from the weaker notion of an \textit{additive} error approximate $t$-design (which is secure only against adversaries making parallel queries) and the stronger notion of \textit{relative} error (which is indistinguishable from Haar-random even for unphysical measurements).\ We will not need these distinctions, so we do not disambiguate. 
\begin{definition}[(Strong) Pseudorandom Random Unitaries]
\normalfont
    A family of unitaries $\mathcal{U} = \{U_{n, k}\}_{n, k}$ is \emph{(strongly) $\eps$-secure against $t$-time quantum adversaries} if for all $n$ and for any quantum algorithm $Q$ running in time $t$,
    \begin{equation*}
        \abs{\underset{k \sim \mathcal{K}}{\text{Pr}}[Q^{U_{n, k},\; (U^\dagger_{n, k})}(1^n) = 1] - \underset{U \sim \mu_N}{\text{Pr}}[Q^{U, \; (U^\dagger)}(1^n) = 1]} \leq \eps
    \end{equation*}
    When $\eps$ is a negligible in $n$ and $t$ is superpolynomial in $n$, we say that $\mathcal{U}$ is a (strong) \textit{pseudorandom unitary}~(PRU).
\end{definition}
In the paper, we take care to distinguish the error of a $t$-design or PRU, from the error incurred by a circuit implementation of such unitaries. For example, we will refer to ``$\eps_0$-approximate $t$-designs, which are $\eps_1$-implemented by a circuit $C$''.  

In this paper, we make use of a particular PRU construction, the $\mathit{CPFC}$ ensemble \cite{metger2024simpleconstructionslineardepthtdesigns, ma2024constructrandomunitaries}, which is the product distribution over the following ensembles of $n$-qubit unitaries: $\mathit{P}$, a uniformly random permutation $p$ on computational basis states, mapping $\ket{x} \mapsto \ket{p(x)}$; $\mathit{F}$, a random Boolean function implemented as a phase oracle $\ket{x} \mapsto (-1)^{f(x)}\ket{x}$; and $\mathit{C}$, the uniform distribution over $n$-qubit Clifford operators\footnote{Or any unitary $2$-design.}. The $C$ on the left and right are independent copies. \ These ensembles are statistically close to the Haar distribution, but can be efficiently instantiated as either $t$-designs or PRUs by replacing the $P$ and $F$ with pseudorandom components. 

\begin{theorem}[\cite{ma2024constructrandomunitaries}]
    \label{thm:cpfc}
   The CPFC ensemble is a \textbf{strong} $\eps(n)$-approximate $t$-design, with $\eps(n) = O({t^2/2^{n/8}})$.
\end{theorem}
Finally, we will need the notion of a $t$-wise independent family of functions:
\begin{definition}[$t$-wise Independent Functions]
    A distribution $\mathcal{D}$ over functions of $\{0,1\}^n$ is \emph{$\eps$-approximate $t$-wise independent} if, for every set of $t$ inputs $x_1, \dots, x_t \in \{0,1\}^n$, the joint distribution of $(f(x_1), \dots, f(x_t))$ for $f \sim \mathcal{D}$ is close to uniform over all $t$-tuples of elements in $\{0,1\}^n$; that is,
    \[
    \norm{(f(x_1), \dots, f(x_t)) - U_D(t))}_1 \le \eps,
    \]
    where $U_D(t)$ denotes the uniform distribution over $t$-tuples with elements in $\{0, 1\}^n$.
\end{definition}
$t$-wise independent functions are useful for constructing random unitaries, as one may imagine from \Cref{thm:cpfc}. 

\section{Random Unitaries in Constant (Quantum) Time}
\label{sec:random-unitaries-in-qac0}
In this section, we analyze explicit constructions of the random unitary ensemble in \Cref{thm:cpfc}. We show that each component of these constructions---the Clifford operator $C$, the phase oracle $F$, and the permutation $P$---can be instantiated by $\QAC^0_f$ circuits, which can then be bootstrapped into $\QAC^0$ circuits. The latter step uses a recent technique for constructing random unitaries known as the \textit{Gluing Lemma} \cite{schuster2025randomunitariesextremelylow}.\ To our knowledge, these are the first constructions of constant-depth unitary $t$-designs and PRUs with standard many-qubit gates \footnote{As this work was finished, Zhang, Vijay, Gu, and Bao \cite{zhang2025designsmagicaugmentedcliffordcircuits} gave a construction of unitary $t$-designs using depth using a Clifford operator followed by a depth $2^{O(t \log t)}$ circuit. In our constructions, we maintain constant depth, at the cost of increasing the number of ancillae.}. 
\subsection{Random Unitaries in \texorpdfstring{$\QAC^0_f$}{}}
First, we observe that both $t$-designs and pseudorandom unitaries can be implemented in $\QAC^0_f$.\ To do so, it suffices to show that suitable derandomizations of $\mathit{P}$ and $\mathit{F}$, as well as $\mathit{C}$, have $\QAC^0_f$ circuits. 
\subsubsection{Unitary Designs}
To show that unitary designs can be implemented in $\QAC^0_f$, we first consider implementing the Clifford operator $\mathit{C}$.\ A folklore construction from measurement-based quantum computing implements any Clifford circuit in constant depth, using mid-circuit measurements and classical feedforward. This construction can be easily translated into a $\QAC^0_f$ circuit. For completeness, we provide a proof in \Cref{app:Clifford}.
\begin{lemma}
    \label{lem:C}
    Let $C$ be any $n$-qubit Clifford Circuit.\ $C$ is implemented by a $\QAC^0_f$ circuit with $O(n^3)$ ancillae.
\end{lemma}
Next, we consider the ensembles of phase oracles $\mathit{F}$.\ Just like random unitaries, a random boolean function requires exponentially many gates to implement, so suitable derandomizations must be chosen in order to obtain efficient circuits.\ For $t$-query adversaries, the following theorem from Zhandry \cite{doi:10.1142/S0219749915500148} allows us to substitute a $2t$-wise independent function family for the oracle in $\mathit{F}$ in place of a truly random function:  
\begin{lemma}[\cite{doi:10.1142/S0219749915500148}]
    \label{lem:zhandry}
    Let $Q$ be any quantum algorithm making at most $t$ queries to $\mathit{F}$. The behavior of $Q$ is unchanged if we replace $\mathit{F}$ with a $2t$-wise independent function family (of phase oracles) $\overline{\mathit{F}}$. 
\end{lemma}
Moreover, $\overline{\mathit{F}}$ can be instantiated in $\TC^0$ due to the following Theorem of Healy and Viola \cite{10.1007/11672142_55}, who gave $\TC^0$ circuits for a well-known family of $t$-wise independent functions:
\begin{lemma}[\cite{10.1007/11672142_55}]
    \label{lem:finite_field}
    Let $\mathbb{F}_{2^s}$ be the finite field with $2^s$ elements. For any $t <2^s$, the following is a family of $t$-wise independent functions: 
    $$\overline{\mathit{F}} = \left\{f(x) = \left(\sum_{i=0}^{t-1} a_ix^i\right): (a_0, \dots, a_{t-1}) \in \mathbb{F}_{2^s}^t \right\} $$
    Furthermore, any function $f: \{0, 1\}^s \rightarrow \{0, 1\}$ from this family can be implemented by a $\TC^0$ circuit with $\Poly(s, t)$ gates.
\end{lemma}
Finally, we consider the random permutation $P$. To our knowledge, there are no known constructions of (even approximate) $t$-wise independent permutations in $\TC^0$ for $t \in \omega(1)$ (for $t = O(1)$, the $2$-round substitution-permutation network of Liu et al.\ \cite{cryptoeprint:2024/083} suffices).\ Fortunately, we can use the Luby-Rackoff inspired construction of Cui, Schuster, Brandão, and Huang \cite{cui2025unitarydesignsnearlyoptimal}, which proves that the permutation $\overline{P}= S_LS_R$, is statistically indistinguishable from random up to inverse exponential error, even when inverse queries are allowed\footnote{As noted in \cite{cui2025unitarydesignsnearlyoptimal}, this will be further analyzed in an upcoming follow-up work.}) when applied together with the $F$ and $C$ ensembles above. $S_L$ and $S_R$ are defined as follows:  
\begin{equation*}
    S_L\ket{x_1 \, \| \, x_2} = \ket{x_1 \oplus f_L(x_2) \, \| \,x_2},\;\; S_R\ket{x_1 \, \| \,x_2} = \ket{x_1 \, \| \, x_2\oplus f_R(x_1)}
\end{equation*}
where $f_L$ and $f_R$ are random functions on $n/2$ bits. Again using \Cref{lem:zhandry}, we can replace $f_L$ and $f_R$ with $2t$-wise independent functions using \Cref{lem:finite_field}. Computing the required \texttt{XOR} of the first half of the input into the second can be implemented in $\AC^0 \subset \TC^0$, so $\overline{P}$ is implementable in $\TC^0$ as desired. Putting together the previous claims in this section with \Cref{thm:tc0_power}, we conclude that any unitary from $C\overline{P} \overline{F}C$ can be implemented in $\QAC^0_f$. Therefore,
\begin{theorem}
    \label{thm:fdesign}
   Strong $\eps$-approximate unitary $t$-designs can be implemented by $\Poly(n, t)$-size $\QAC^0_f$ circuits, with $\eps = \exp(-{\Omega(n)})$.
\end{theorem}
\subsubsection{Pseudorandom Unitaries} \label{sec:pseudorandom_units}
To instantiate pseudorandom unitaries in $\QAC^0_f$, we can repeat the analysis for unitary $t$-designs, using a quantum-secure pseudorandom function $\mathit{F^*}$ used instead of a $2t$-wise independent function.\ A natural candidate $F^*$ is the Ring-LWE pseudorandom function proposed by Banerjee, Peikert, and Rosen~\cite{cryptoeprint:2011/401} and further analyzed by Zhandry \cite{10.1145/3450745}. This (candidate) quantum-secure pseudorandom function is keyed by a matrix $A \in \mathbb{Z}_{p}^{m \times n}$ and $l$ different $n \times n$ matrices $\{S_i\}_{i=1}^l$ over $\mathbb{Z}_q$, and acts as follows:
\begin{equation}
    \label{eq:lweprf}
    \text{PRF}_{(\mathbf{A}, \{S_i\}_{i=1}^l)}(x) = \left[ \mathbf{A}^t \prod_{i=1}^{\ell} S_i^{x_i} \right]_p
\end{equation}
For certain reasonable choices of $p,q,m,l$,~\cite{cryptoeprint:2011/401} show that this PRF is secure and assuming that the Learning With Errors (LWE)~\cite{regev2024latticeslearningerrorsrandom} problem is post-quantum secure, a standard assumption in post-quantum cryptography. Moreover, ~\cite{cryptoeprint:2011/401} also show that \Cref{eq:lweprf} can be typecast into a boolean function by using standard rounding techniques, while maintaining precision guarantees. Moreover, they also prove that this function can be implemented in $\TC^0$:
\begin{lemma}[\cite{cryptoeprint:2011/401, 10.1145/3450745}]
      Assuming post-quantum security of LWE The function in \Cref{eq:lweprf} is $\eps_0(n)$-secure against $t(n)$-time for quantum adversaries, the above function family is $\eps_0(n)/\Poly(n)$-secure against $t(n)$-time adversaries.\ Here, $\eps_0(n)$ is a negligible function which depends the strength of the LWE assumption, and $t(n)$ is superpolynomial in $n$. 
\end{lemma}
Let $P^*$ be defined analogously to $\overline{P}$ except with $F^*$ in place of $\overline{F}$. 
Invoking the triangle inequality, the distinguishing advantage of any adversary that queries $CP^*F^*C$ in place of a Haar-random unitary is at most the sum of the distinguishing advantage of $F^*$,\ $P^*$,\ and the error from \Cref{thm:cpfc}:
\begin{lemma}
    Suppose that $\mathit{F^*}$ is $\eps_0(n)$ secure against $t(n)$-time adversaries. Then, the ensemble $\mathit{CP^*F^*C}$ is strongly $O(\eps_0(n)) + \eps(n))$ secure against $t(n)$-time adversaries, for $\eps(n) = \Poly(t(n))/\exp(-\Omega(n))$. 
\end{lemma}
Putting together the above claims again with \Cref{thm:tc0_power}, we conclude the following: 
\begin{theorem}
    \label{thm:fpru}
    Suppose the Learning With Errors (LWE) problem is $\eps_0(n)$-secure against $t(n)$-time for quantum adversaries. Then, there exists an ensemble of unitaries implementable by $\Poly(n)$-size $\QAC^0_f$ circuits, which is strongly $\Theta(\eps_0(n)/\Poly(n) + \Poly(t(n))/\exp(-\Omega(n)))$-secure against $t(n)$-time adversaries. 
\end{theorem}
Taking $t$ to be both subexponential and superpolynomial in $n$, and $\eps_0$ to be any negligible function of $n$, we obtain an implementation of strong PRUs by $\QAC^0_f$ circuits. Moreover, assuming subexponential post-quantum security of LWE, i.e. that there exists some $c >0$ such that $\eps_0(n) = 1/2^{n^c}$-security holds all $t(n) = 2^{n^c}$-time algorithms, we obtain subexponentially secure strong PRUs by plugging in the relevant parameters. 
\begin{corollary}
    \label{cor:subexp_pru}
    Suppose that LWE has subexponential post-quantum security\footnote{This was the assumption in \cite{schuster2025randomunitariesextremelylow}.}. Then, there exists an ensemble of unitaries implementable by $\Poly(n)$-size $\QAC^0_f$ circuits, which is strongly $\Theta(1/2^{n^{c^\prime}})$-secure against $2^{n^{c^\prime}}$-time adversaries, for some $c^\prime \in (0, 1)$.
\end{corollary}
We use $c^\prime$ to emphasize that this constant may differ from the constant $c$ in the preceding paragraph. 

Similar to the case of unitary designs, directly implementing (quantum-secure) pseudorandom permutations in $\QAC^0_f$ is challenging, as typical classical constructions for pseudorandom permutations fail in the quantum setting \cite{zhandry2016notequantumsecureprps, kuwakado}, so we rely on the Luby-Rackoff inspired construction to implement the permutation.

\subsection{Random Unitaries in \texorpdfstring{$\QAC^0$}{}}
To bootstrap constructions of (pseudo)random unitaries in $\QAC^0_f$ to $\QAC^0$, we will leverage the recently discovered ``Gluing Lemma", proved by Schuster, Haferkamp, and Huang~\cite{schuster2025randomunitariesextremelylow}.\ The Gluing Lemma is an extremely powerful statement, which allows random unitaries on $n$ qubits to be ``glued'' from several copies of smaller $O(\text{poly}\log(n))$ sized unitaries\footnote{A similar ``gluing'' result was proved by LaRacuente and Leditzky \cite{laracuente2024approximateunitarykdesignsshallow} to obtain low-depth unitary $t$-designs. For our purposes, it will be more convenient to use the constructions in \cite{schuster2025randomunitariesextremelylow}.}—see \Cref{fig:gluing} for an illustration. In particular, we have the two following statements, taken from and proved in~\cite{schuster2025randomunitariesextremelylow}:
\begin{theorem}[Gluing Unitary $t$-designs, \cite{schuster2025randomunitariesextremelylow} Appendix B.3.]
\label{thm:gluingdesigns}
    Consider a two layer, $n$-qubit brickwork circuit with local patch size $\ell$. Suppose each gate in the brickwork is independently drawn from an $\eps/n$-approximate unitary $t$-design. Then, as long as $\ell \ge \Theta(\log(nt^2/\eps))$, the overall ensemble forms an $\eps$-approximate unitary $t$-design on $n$ qubits. 
\end{theorem}
\begin{theorem}[Gluing Pseudorandom Unitaries, \cite{schuster2025randomunitariesextremelylow} Appendix C.4.]
\label{thm:gluingprus}
    Consider a two layer, $n$-qubit brickwork circuit with local patch size $\ell$. Suppose each $\ell$-qubit gate in the brickwork is independently drawn from an ensemble which is $\eps(\ell)$-secure against $t(\ell)$-time adversaries.\ Then, as long as $\ell \ge \omega(\log n)$, the overall ensemble is $O(n) \cdot \eps(\ell)$-secure against $t(\ell)$-time adversaries. 
\end{theorem}

\begin{figure}[ht]
    \includegraphics[width=0.9\linewidth]{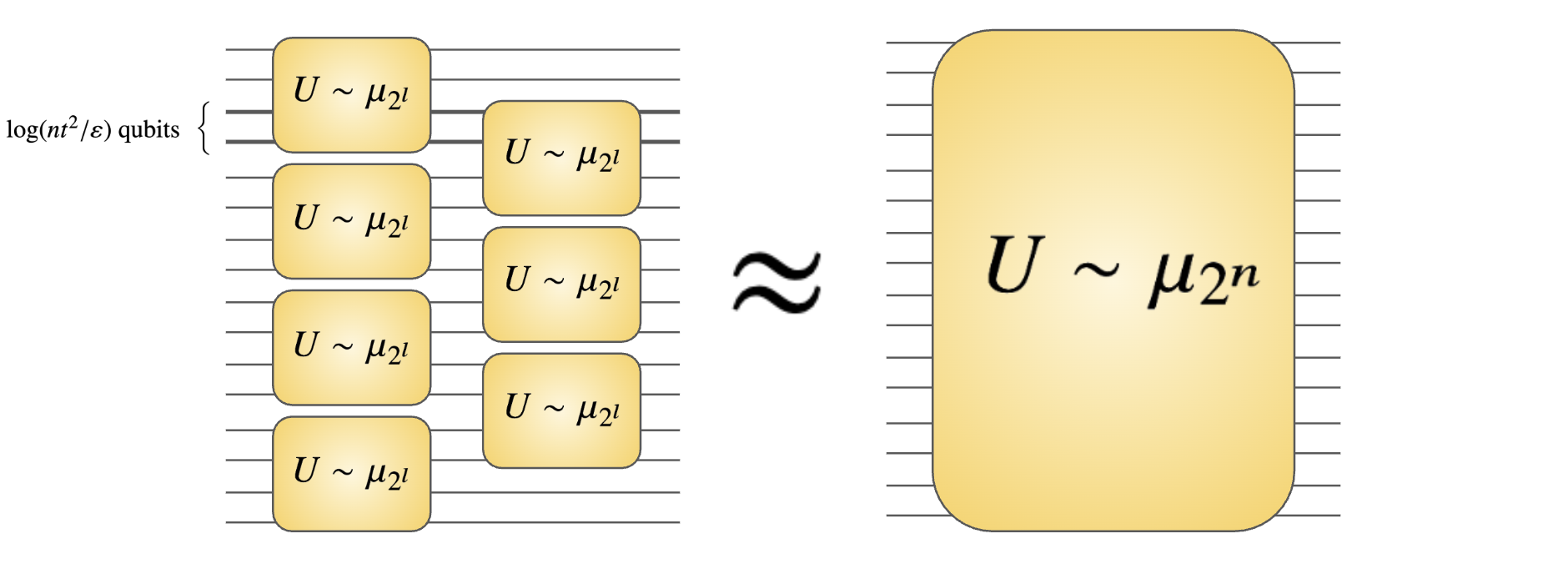}
    \caption{The Gluing Lemma of Schuster, Haferkamp, and Huang \cite{schuster2025randomunitariesextremelylow}} gives a recipe for constructing $n$-qubit random unitaries by ``gluing'' together copies of $l << n$ qubit random unitaries.\ The size $\ell$ of the local patch depends on the desired parameters of the random unitary to be constructed. 
    \label{fig:gluing}
\end{figure}
Since the brickwork circuit is only two layers, if the individual unitaries are implementable by $\QAC^0_f$ circuits, then the whole circuit is as well. Combining these with \Cref{thm:fdesign} and \Cref{thm:fpru}, we can exponentially reduce the width of the many-qubit gates required to form random unitaries (in our $\QAC^0_f$ constructions, the \texttt{FANOUT} gates act on $\Poly(n)$ qubits):
\begin{lemma}
    \label{lem:glueddesigns}
    $\eps$-approximate unitary $t$-designs can be implemented in $\QAC^0_f$, where every $\texttt{FANOUT}$ or $\texttt{TOFFOLI}$ gate acts on at most $\Poly(\log(nt^2/\eps)) \cdot \Poly(t)$ qubits.  
\end{lemma}
We have an analogous theorem for PRUs, combining \Cref{cor:subexp_pru} and \Cref{thm:gluingprus}:
\begin{lemma}
    \label{lem:glueprus}
    Let $k$ be any constant. Suppose that LWE has subexponential post-quantum security. Taking a local patch size of $\ell = \log^{k/c}n$ in the statement of \Cref{thm:gluingprus}, we obtain an ensemble of unitaries which is $O(n) \cdot 1/2^{\log^{1/c}n}$-secure against $2^{\log^{k}n}$-time adversaries. Furthermore, every $\texttt{FANOUT}$ or $\texttt{TOFFOLI}$ gate acts on at most $\Poly \log^{k}n$ qubits. 
\end{lemma}
Now that the width of the many-qubit gates has been reduced, our general strategy will be to substitute every $\texttt{FANOUT}$ gate with an approximating $\QAC^0$ circuit.

Rosenthal proved that \texttt{PARITY} (and therefore also \texttt{FANOUT}) is equivalent (up to a $\QAC^0$ reduction to preparing a $n$-\textit{nekomata} state, or a state of the form $\frac{\ket{0^n, \psi_0} + \ket{1^n, \psi_1}}{\sqrt{2}}$. A state is called an \textit{$\eps$-approximate $n$-nekomata} if it has fidelity at least $1 - \eps$ with some $n$-nekomata. In particular, Rosenthal showed the following:
\begin{lemma}
    \label{lem:neko}
    Let $C$ be a circuit which constructs an $\eps$-approximate nekomata from the all-zero state. Then, there exists a $\QAC^0$ circuit, which makes queries to $C$ and $C^\dagger$, which $1 - O(\eps)$-implements the \texttt{PARITY} \textit{unitary} operation.
\end{lemma}
Rosenthal also showed that exponential-sized $\QAC^0$ circuits exist, giving the first $\mu$-implementation\footnote{We use $\mu$ to distinguish from the $t$-design approximation error $\eps$.} of \texttt{PARITY} in $\QAC^0$ with exponentially many ancillae. Rosenthal's original construction was further analyzed by Grier and Morris~\cite{grier2024quantumthresholdpowerful}, whose result we will use in what follows.
\begin{fact}[\cite{grier2024quantumthresholdpowerful}]
    \label{fact:grier}
    A $w$-qubit nekomata gate can be $\mu$-approximated, with $\mu = 1/2^{\Theta(w)}$, by a constant-depth $\QAC^0$ circuit $C$ that uses $\exp(\Theta(w))$ many ancillae. Therefore, a width-$w$ \texttt{PARITY} unitary operator can be $ 1/2^{\Theta(w)}$-implemented by a $\QAC^0$ circuit with $\exp(\Theta(w))$ ancillae. 
\end{fact}
This construction can be slightly improved in two ways while remaining in $\QAC^0$. First, we observe that constant-sample amplification can be performed in $\QAC^0$:
We observe that this construction can be slightly improved in $\QAC^0$:
\begin{corollary}
    \label{cor:first_par_improvement}
    For any constant $d$, a width-$w$ \texttt{PARITY} unitary operator can be $ 1/2^{\Theta(dw)}$-implemented by a $O(\log d)$-depth $\QAC^0$ circuit with $d \cdot \exp(\Theta(w))$ ancillae. 
\end{corollary}
\begin{proof}
    Since $\QAC^0$ contains reversible \texttt{AND} and \texttt{OR} gates with constant fan-in, $d$-qubit \texttt{FANOUT} and \texttt{MAJORITY} gates can be exactly implemented by $\QAC^0$ circuits with $O(\log(d))$ ancillae\footnote{\url{https://cstheory.stackexchange.com/questions/21386/circuit-complexity-of-majority-function}}. Therefore, as a preprocessing step to the Morris-Grier constructions, we first apply a width-$O(d)$ \texttt{FANOUT} gate to the input state. On each copy, we apply the Grier-Morris construction. Thus, we obtain $O(d)$ independent \texttt{PARITY} computations, each correct with probability $1/2^{\Theta(w)}$. Finally, we can take a width-$O(d)$ \texttt{MAJORITY} of the outputs to obtain the final output of the circuit, then uncompute the intermediate steps as necessary. By a standard Chernoff bound, this decreases the error in the \texttt{PARITY} computation from $1/2^{\Theta(w)}$ to $1/2^{\Theta(dw)}$.
\end{proof}
Second, we use the downward self-reducibility of \texttt{PARITY}: 
\begin{figure}[ht]
\centering
\begin{tikzpicture}[scale=0.8, transform shape]
\node (tree1) at (0,0) {
\begin{forest}
for tree={
  if n children=0{
    content={$\vdots$},
    no edge
  }{
    mynode,
    edge={very thick, -{Stealth}},
    inner sep=1pt,
    s sep=12mm,
    l=6mm,
    content={$\texttt{PARITY}_{w}$}
  },
  grow'=south,
  anchor=center,
  parent anchor=south,
  child anchor=north
}
[$\texttt{PARITY}_{w}$
  [$\texttt{PARITY}_{w}$
    [$\texttt{PARITY}_{w}$
      [$\vdots$]
      [$\vdots$]
    ]
    [$\texttt{PARITY}_{w}$
      [$\vdots$]
      [$\vdots$]
    ]
  ]
  [$\texttt{PARITY}_{w}$
    [$\texttt{PARITY}_{w}$
      [$\vdots$]
      [$\vdots$]
    ]
    [$\texttt{PARITY}_{w}$
      [$\vdots$]
      [$\vdots$]
    ]
  ]
]
\end{forest}
};

\node (equals) at (6,0) {\huge $=$};

\node[mynode] (tree2) at (9,0) {$\texttt{PARITY}_{w^d}$};
\end{tikzpicture}

\caption{\texttt{PARITY} is downward self-reducible: a depth-$d$ tree with fan-in \( w \) at each layer computes \texttt{PARITY} on $w^d$ qubits.}
\label{fig:par_tree}
\end{figure}
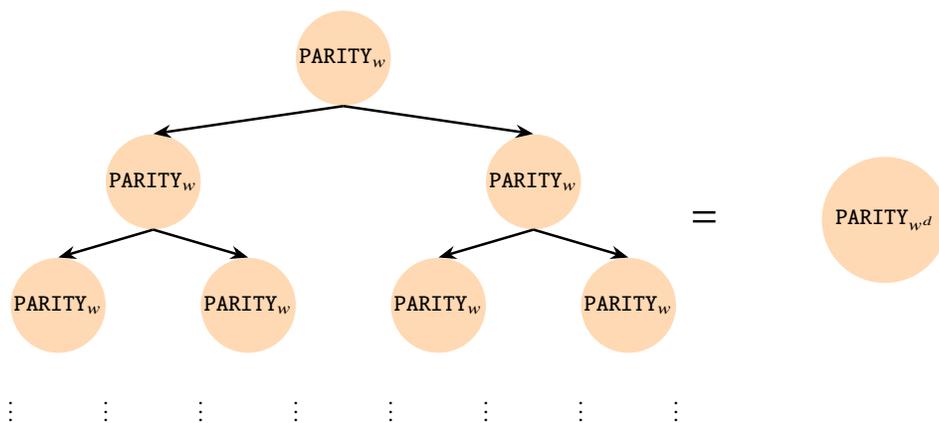
\begin{corollary}
    \label{cor:parity_improved}
    For any constant $d$, a width-$w^d$ \texttt{PARITY} unitary operator can be $ w^d/2^{\Theta(dw)}$-implemented by a $O(d\log d)$-depth $\QAC^0$ circuit with $dw^d \cdot \exp(\Theta(w))$ ancillae. 
\end{corollary}
\begin{proof}
    \texttt{PARITY} is downward self-reducible: a width-$w^d$ \texttt{PARITY} gate can be implemented by a depth-$d$ tree where each leaf node corresponds to $w^d$ qubits, and internal nodes which compute the \texttt{PARITY} of their $w$ children (\Cref{fig:par_tree}). This tree has $O(w^d)$ nodes. If each internal node implements \texttt{PARITY} with error $\mu$, then by \Cref{lem:compose} and \Cref{fact:grier} the total error is at most $O(w^d\mu)$. 

Since all internal nodes can be uncomputed after the root node computes the overall parity, the total depth increases by a multiplicative factor of $2d$, and each of the nodes uses at most $d \cdot \exp(O(w))$ ancillae by \Cref{cor:first_par_improvement}. Thus, we obtain a smaller circuit for implementing \texttt{PARITY} at the cost of increasing the depth and error parameters.
\end{proof}
Since \texttt{PARITY} and \texttt{FANOUT} are Hadamard conjugates (\Cref{fact:par_fanout}), \Cref{cor:parity_improved} also holds for \texttt{FANOUT}. 

\subsubsection{Unitary Designs}
To obtain unitary $t$-designs in $\QAC^0$, we replace each \texttt{FANOUT} gate in the ``glued unitary'' from \Cref{lem:glueddesigns} by an appropriate $\mu$-implementation:
\begin{theorem}
    \label{thm:qac0designs}
    For any $\eps > 0$ and $c_0 \in \mathbb{N}$, $\eps$-approximate $n$-qubit unitary $t$-designs can be $\mu$-approximately implemented by size $s$, depth $O(c_1 \log c_1)$ $\QAC^0$ circuits, with  
    \begin{equation*}
        \mu = 1/\exp\left({\Theta\left(\frac{c_1t}{c_0}\log(nt^2/\eps)\right)}\right) \;\; s = \exp\left(\frac{t}{c_0}\log(nt^2/\eps)\right)
    \end{equation*}
    Here, $c_1$ is any sufficiently large constant. 
\end{theorem}
\begin{proof}
    \Cref{lem:glueddesigns} allows us to construct an $\eps$-approximate $t$-design using by implementing patches of smaller $\QAC^0_f$ unitaries each acting on $\ell = \Poly(\log(nt^2/\eps), t)$ qubits. Therefore, each \texttt{FANOUT} gate acts on at most $\ell$ qubits as well. 
    
    Let $c_0 \in \mathbb{N}$, and $w = \frac{t}{c_0}\log(nt^2/\eps)$. Let $c_1 \in \mathbb{N}$ be such that $w^{c_1} \ge \ell$—such a $c_1$ exists by \Cref{lem:glueddesigns}). \Cref{cor:parity_improved} implies that any \texttt{FANOUT} gate with width at most $\ell$ can be 
    \begin{equation*}
        w^{c_1}/2^{\Theta(c_1w)} = 1/\exp\left({\Theta\left(\frac{c_1t}{c_0}\log(nt^2/\eps)\right)}\right)
    \end{equation*}
    implemented by a $\QAC^0$ circuit with 
    \begin{equation}
        \label{eq:design_one}
        c_1w^{c_1} \cdot \exp(w) = \exp\left(\frac{t}{c_0}\log(nt^2/\eps)\right)
    \end{equation}
     ancillae. Note that we have absorbed lower-order terms into $\Theta(\cdot)$ and $\exp(\cdot)$.
    
    Since the original $\QAC^0_f$ circuit from \Cref{lem:glueddesigns} has at most $n \ell$ \texttt{FANOUT} gates,  \Cref{lem:compose} implies that the overall implementation error is at most
    \begin{equation}
        \label{eq:design_two}
        \frac{n\ell}{\exp\left({\Theta\left(\frac{c_1t}{c_0}\log(nt^2/\eps)\right)}\right)} = \frac{1}{\exp\left({\Theta\left(\frac{c_1t}{c_0}\log(nt^2/\eps)\right)}\right)}
    \end{equation}
     where we again absorbed lower-order terms into $\exp(\cdot)$. Similarly, the total number of ancillae required is at most $n\ell$ times the expression in \Cref{eq:design_one}.   
\end{proof}

  A particularly important parameter regime is when $t$ is polynomial in $n$, and $\eps$ is inverse polynomial in $n$. In this regime, we can $\mu$-implement an $\eps$-approximate $t$-design where $\mu$ is any inverse polynomial in $n$ (by choosing $c$ sufficiently large in \Cref{eq:design_two}). This construction requires $n \Poly\log n \cdot n^{\delta t}$ ancillae, for any constant $\delta > 0$ (by choosing $c_0$ sufficiently large in \Cref{eq:design_one}). Hence, we obtain the following corollary:
  \begin{corollary}
        \label{cor:good_designs}
      For any $\delta > 0$, and sufficiently large $n$ depending on $\delta$, there exists a $\QAC^0$ circuit with $o(n^{1 + \delta})$ ancillas which $\mu$-implements an $\eps$-approximate $O(1)$-design, where $\mu$ and $\eps$ are inverse polynomial in $n$.
  \end{corollary}
 Note that the circuit depth scales as $O(1/\delta)$. Given that \texttt{PARITY} cannot be computed by depth-$d$ $\QAC^0$ circuits with $n^{1 + 1/3^d}$ ancillae~\cite{anshu2024computationalpowerqac0barely}, \Cref{cor:good_designs} implies that with slightly more ancillae ($n^{1 + 1/O(d)}$), we already obtain good approximations of constant $t$-designs (as long as $n$ is sufficiently large). This observation will also hold for the pseudorandom unitary constructions in the next section. 
  
 Moreover, $\eps$ and $\mu$ can be both be made inverse-\textit{quasipolynomial} in $n$, although the resulting constructions will require quasipolynomial ancillae. More generally, obtaining any construction with $o(1/\eps)$ size would likely require improving the tradeoff between size and implementation error in either the Rosenthal or Morris-Grier constructions, which we leave as an interesting open problem. 
\subsubsection{Pseudorandom Unitaries}
We can now implement PRUs in $\QAC^0$ in an analogous fashion. 
\begin{theorem} 
    \label{thm:qac0prus}
   For any $k, c_0 \in \mathbb{N}$, assuming subexponential post-quantum security of LWE, PRUs secure against $n^{\Poly \log n}$-time adversaries can be $\mu$-implemented by size $s$, depth $O(c \log c)$ $\QAC^0$ circuits, with  
    \begin{equation*}
        \mu = 1/\exp(\Theta(c\log^k n/c_0)) \;\; s = \exp(\log^k n/c_0)
    \end{equation*}
    Again, $c$ is any sufficiently large constant (where ``sufficiently large'' depends on $c_0$).
\end{theorem}
\begin{proof}
    We reuse the analysis of \Cref{thm:qac0prus}, with $w = \log^k n/c_0$ and local patch size $\ell = w^c$ for some constant $c$ (the $\ell$ from \Cref{lem:glueddesigns}), obtaining a $1/\exp(\Theta(c\log n/c_0))$-implementation with $\exp(\log n/c_0)$ ancillae. 
\end{proof}
Taking $a = 1$, we obtain implementations of PRUs by polynomial-size $\QAC^0$ circuits. Like \Cref{cor:good_designs}, the polynomial can be chosen to be $o(n^{1 + \delta})$ for any $\delta > 0$, almost linear when $\log n \gg c_0$. When $k > 1$, we implement PRUs to inverse \textit{quasipolynomial} error, although the resulting circuits have quasipolynomial size. Ideally, one would hope to implement a PRU to negligible error in polynomial size $\QAC^0$. Unfortunately, we are again bottlenecked by the $\Theta(1/\eps)$ scaling in \Cref{fact:grier}. \\ \\
One important caveat implicit in both $\QAC^0$ constructions is that the desired random unitary on the $n$ input qubits is only approximately implemented, and there is no guarantee that the action when tracing out the ancilla qubits is unitary\footnote{In fact, it is an open problem to exactly implement \texttt{FANOUT} using $\QAC^0$ circuits of any size.}. In many applications, this scenario is addressed by giving the distinguishing algorithm access to the \textit{channel} $\Phi_C(\rho) \coloneqq \tr_{anc}(C(\rho \otimes \ket{0^a}\bra{0^a})C^\dagger)$, which maps the state of the $n$ input qubits to the reduced state on the same $n$ qubits after applying $C$.

If $C$ $\mu$-implements a unitary $U$ (on $n$ qubits), then replacing queries to $U$ with queries to $\Phi_C$ incurs an additive error:
\begin{lemma}
    \label{lem:approximator}
    Assume $C$ $\eps$-approximates a unitary $U$.\ Let $Q$ be any quantum algorithm which makes at most $t$ queries to either $\Phi_C/\Phi_{C^\dagger}$ or $U/U^\dagger$. For any input state $\rho$, $$\norm{Q^{\Phi_C, \Phi_{C^\dagger}}\rho(Q^{\Phi_C, \Phi_{C^\dagger}})^\dagger - Q^{U, U^\dagger}\rho(Q^{U, U^\dagger})^\dagger}_{td} \le t{\mu}$$
\end{lemma}
\begin{proof}
    We assume without loss of generality that all queries are to $U$: the proof is identical if queries to $U^\dagger$ are replaced by queries to $\Phi_{C^\dagger}$. Define $Q_0 = Q^{\Phi_C}, Q_t = Q^U$, and $Q_i$ the algorithm whose first $i$ queries are to $U$ and whose last $t - i$ queries are to $\Phi_C$. By the triangle inequality, this is at most 
    \begin{align*}
        t \cdot \sum_{i=0}^{t-1} \norm{Q_{i + 1}\rho(Q_{i + 1})^\dagger - Q_i\rho(Q_i)^\dagger}_{td} \le t \cdot \sum_{i=0}^{t-1} \norm{\mathcal{E}_i(U\rho_i U^\dagger) -  \mathcal{E}_i(\tr_{anc}(C(\rho_i \otimes \ket{0^a}\bra{0^a})C^\dagger))}_{td}  \\
        \le t \cdot \sum_{i=0}^{t-1} \norm{U\rho_i U^\dagger - \tr_{anc}(C(\rho_i \otimes \ket{0^a}\bra{0^a})C^\dagger)}_{td} 
    \end{align*}
    for some channels $\mathcal{E}_i$ and states $\rho_i$, where in the second line we use that trace distance can only decrease under an application of $\mathcal{E}_i$. Writing $\rho_i$ as a linear combination $\rho_i = \sum_j \alpha_{ij} \ket{\psi_j}\bra{\psi_j}$,
    \begin{align*}
        t \cdot \sum_{i=0}^{t-1} \sum_j \alpha_{ij} \norm{U \ket{\psi_j}\bra{\psi_j}U^\dagger - \tr_{anc}(C(\ket{\psi_j}\bra{\psi_j} \otimes \ket{0^a}\bra{0^a})C^\dagger)}_{td} 
        \\ \le t \cdot \sum_{i=0}^{t-1} \sum_j \alpha_{ij} \norm{U \ket{\psi_j}\bra{\psi_j}U^\dagger \otimes \ket{0^a}\bra{0^a} - C(\ket{\psi_j}\bra{\psi_j} \otimes \ket{0^a}\bra{0^a})C^\dagger}_{td} \\
        \le t \cdot \sum_{i=0}^{t-1} \sum_j \alpha_{ij}  (1 -\sqrt{1 - \mu}) = t \cdot (1 -\sqrt{1 - \mu}) \le t\mu
    \end{align*}
    from the first to the second line, we used that trace distance is non-increasing when taking a partial trace. From the second to third line, we use the Fuchs-van-de Graff inequality and the definition of $\mu$-implementation.
\end{proof}

Since the maximum distinguishing advantage between two quantum states is equal to the trace distance, combining \Cref{thm:qac0prus} and \Cref{lem:approximator} implies that any $n^{\Poly \log n}$-time algorithm which makes at most $t$ queries to a polynomial size $\QAC^0$ circuit $C$ cannot, in the worst case distinguish $C$ from Haar-random except with at most $t$ times that probability. In particular, we obtain the following distinguishing lower bound: 
\begin{theorem}
    \label{thm:lowerbound}
    Assuming subexponential post-quantum security of LWE, no fixed polynomial-time algorithm can distinguish the following two cases with non-negligible probability:
    \begin{itemize}
        \item $C$ is a polynomial size $\QAC^0$ circuit.
        \item $\Phi_C = \Phi_U = U\rho U^\dagger$ for a Haar-random unitary $U$. 
    \end{itemize}
\end{theorem} 
\begin{proof}
    Assume by contradiction that a $O(n^{c_0})$-time algorithm $A$ could distinguish any polynomial-size $\QAC^0$ circuit from random, with probability $\Omega(1/n^{c_1})$. Use \Cref{thm:qac0prus} to implement a PRU to error $O(1/n^{c_2})$, where $c_2 \gg c_0 + c_1$ in polynomial-size $\QAC^0$. Then, \Cref{lem:approximator} implies that $A$'s maximum distinguishing advantage should be $O(n^{c_0 - c_2}) = o(1/n^{c_1})$, a contradiction. 
\end{proof}

\subsection{Learning Lower-Bounds from Pseudorandomness}
\label{sec:new_approaches}
We will now show how our implementations of PRUs (and quantum secure PRFs as in \Cref{eq:lweprf}) can be used to prove powerful computational lower bounds for average-case learning of $\QAC^0$. 

In particular, we will prove bounds with respect to the average-case distance measure, which is defined for two $n$-qubit CPTP maps $\mathcal{E}_1$ and $\mathcal{E}_2$ as
\begin{align*}
    \mathcal{D}_\text{avg}(\mathcal{E}_1, \mathcal{E}_2)=\underset{\ket{\psi}\sim \mu_{2^n}}{\mathbb{E}}\left[1-F\left(\mathcal{E}_1(\ketbra{\psi}{\psi}),\mathcal{E}_2(\ketbra{\psi}{\psi})\right)\right],
\end{align*}
where $\ket{\psi}$ is sampled from the Haar measure $\mu_{2^n}$ and $F(\cdot, \cdot)$ is the fidelity. In the case of unitary channels, the Haar distance measure simplifies to the average gate fidelity error. As such, we will abuse notation for unitary channels and write $\mathcal{D}_\text{avg}(U_1, U_2)$ to mean the distance according to the corresponding unitary channels. Important to our analysis will be the following fact:
\begin{fact}[\cite{Nielsen_2002}] \label{fact:uni_chans_dist}
    For unitaries $U_1$ and $U_2$, the average-case distance satisfies
    \begin{align}
        \mathcal{D}_\text{avg}(U_1, U_2) = \frac{2^n}{2^n+1} \left(1-\frac{1}{4^n}\left|\tr(U_1^\dag U_2)\right|^2\right).
    \end{align}
\end{fact} 

We begin by using the PRU construction of \Cref{thm:qac0prus} (assuming hardness of Learning with Errors) to prove a super-polynomial computational lower bound for learning $\QAC^0$ circuits with polynomial ancillae.

\begin{lemma}
    \label{thm:learning_PRU}
    Assuming subexponential post-quantum security of LWE, $n^{\omega(1)}$ time is necessary to learn polynomial-sized $\QAC^0$ unitaries, according to the average-case distance measure.  
\end{lemma} 
\begin{proof}
    By \Cref{thm:lowerbound}, we can produce a $1/\Poly(n)$-approximate PRU $U$ in polynomial-sized $\QAC^0$ with $\omega(n^c)$-time security, for arbitrary constant $c=O(1)$. For contradiction, assume that there exists an $O(n^c)$-time algorithm, for learning and implementing a unitary $\widetilde{U}$ such that $\mathcal{D}_\text{avg}(U, \widetilde{U}) \leq 1/\Poly(n)$. Using the quantum circuits for $U$ and $\widetilde{U}$, we will show how to implement a polynomial time, i.e. $O(n^{\widetilde{c}})$-time test that can be used to distinguish $U$ from Haar-random. Since, $c$ can be made arbitrarily large, we can thus pick a PRU with $c > \widetilde{c}$, such that the overall procedure for learning and distinguishing requires time complexity, $O(n^c+n^{\widetilde{c}}) = O(n^c)$. This contradicts the security of the PRU, thus establishing that $\omega(n^{c})$-time is necessary to learn $U$. Furthermore, since we can amplify $c$ to be an arbitrarily large constant, this implies super-polynomial time is necessary to learn all possible $1/\Poly(n)$-precise PRUs to the desired precision.
    
    All that remains is to show that there exists 
    a polynomial time algorithm for distinguishing $U$ from a Haar random unitary, via the learned circuit $\widetilde{U}$. For ease of notation, we will denote $d=2^n$.
    To do so, we convert each unitary into a state via the Choi-Jamiolkowski isomorphism. In particular, let 
    $\ket{\Phi} = \frac{1}{\sqrt{d}} \sum_{i=1}^{d} \ket{i}\otimes \ket{i}$ denote the maximally entangled state, such that 
    $\ket{\psi_U} = (U\otimes I) \ket{\Phi}$ and $\ket{\psi_{\widetilde{U}}} = (\widetilde{U}\otimes I) \ket{\Phi}$ 
    are the respective Choi states of $U$ and $\widetilde{U}$. Performing a swap test between $\ket{\psi_U}$ and $\ket{\psi_{\widetilde{U}}}$ results in the zero outcome with probability
    \begin{align*}
        P_0 &= \frac{1}{2}\left(1+|\braket{ \psi_{\widetilde{U}}|\psi_U}|^2\right)= \frac{1}{2}\left(1+\left|\frac{1}{d}\sum_{i,j}\bra{j}\widetilde{U}^\dagger U \ket{i}\braket{j|i}\right|^2\right) \\
        &= \frac{1}{2}\left(1+\frac{1}{d^2}\left|\sum_{i}\bra{i}\widetilde{U}^\dagger U \ket{i}\right|^2\right) = \frac{1}{2}\left(1+\frac{1}{d^2}|\tr(\widetilde{U}^\dagger U)|^2\right).
    \end{align*}
    In the case that unitary $U$ is a $\QAC^0$ PRU, by the guarantees of the learning algorithm, this implies that
    \begin{align*}
        &\mathcal{D}_\text{avg}(U, \widetilde{U}) = \frac{d}{d+1} \left(1-\frac{1}{d^2}|\tr(\widetilde{U}^\dagger U)|^2\right) \leq \eps \implies  P_0 \geq 1 - \eps \cdot \frac{d+1}{2d}, 
    \end{align*}
    for $\eps = 1/\Poly(n)$, meaning that $P_0\approx 1$. Meanwhile, in the case that $U$ is Haar-random, for any fixed (learned) $\widetilde{U}$, the expectation of $P_0$ is tightly concentrated around 
    \begin{align*}
        \mathbb{E}_{U \sim \mu}\left[P_0\right] = \frac{1}{2}\left(1+\frac{1}{d^2}\mathbb{E}_{U \sim \mu}\left[|\tr(\widetilde{U}^\dagger U)|^2\right] \right)= \frac{1}{2}+\frac{1}{2d^2} \approx \frac{1}{2}.
    \end{align*}
    Therefore, by running the swap test just a constant number of times, we can estimate $\hat{P}_0$ and decide, say, if $\hat{P}_0<3/4$ then $U$ is Haar random (otherwise it is not)---thus distinguishing PRUs from Haar random unitaries.
\end{proof}

Using a similar proof approach and just the PRFs (through which the previous PRUs were constructed), we are also able to extend the lower-bound to the sub-linear ancilla setting (improved over the polynomial ancilla required for the PRU implementation). This proof structure is more directly analogous to the proofs of classical \cite{10.1145/167088.167197} and quantum \cite{doi:10.1137/20M1344202} hardness for average-case learning of $\AC^0$. 
\begin{theorem}
    \label{thm:learning_PRF}
    Fix any constant $\delta > 0$. Assuming subexponential post-quantum security of LWE, $n^{\omega(1)}$ time is necessary to learn $\QAC^0$ unitaries with $O(n^\delta)$ ancillae, according to the average-case distance measure.  
\end{theorem} 
\begin{proof}
    We will leverage the Banerjee, Peikert, and Rosen PRF construction described in \Cref{sec:pseudorandom_units}. Following from \Cref{thm:fpru}, we can implement a PRF in polynomially sized $\QAC^0_f$ that is secure against subexponential adversaries. Specifically, we will assume that the security is for time $t(n) = 2^{n^{1/k}}$, for some constant $k$. In order to efficiently implement the PRF in $\QAC^0$, we can use a procedure similar to that of the gluing lemma (\Cref{lem:glueddesigns}). Namely, we will ``shrink'' the PRF to only act on inputs of size $m=\log^d(n)$, where $d$ is a constant strictly greater than $k$, i.e. $d>k$. Note that we act by identity on the remaining $n-m$ input qubits. The security of this shrunk PRF on is still quasipolynomial, i.e. $t(m) = 2^{\log^{d/k}(n)} = 2^{\Poly\log(n)}$. Furthermore, implementing this PRF to arbitrary polynomial precision can be accomplished by a $\QAC^0$ circuit with $n^\delta$ ancillae, by the same logic as the single-patch step given by \Cref{eq:design_one} in the proof of \Cref{thm:qac0designs}. In particular, we no longer need an $O(n\ell)$ factor accounting for the system size. 

    Now, assume for contradiction that there existed an $O(n^k)$ time algorithm for learning a unitary $\widetilde{U}$ such that $\mathcal{D}_\text{avg}(U, \widetilde{U}) \leq \eps$. Via an analogous Choi state and swap test distinguishing procedure to that of \Cref{thm:learning_PRU}, we can thus distinguish the pseudorandom function from a truly random function. (Note that the only change in the analysis is that $U$ is now a unitary encoding of a $\Poly\log(n)$ size random function, which still has $\mathbb{E}_{U}[\frac{1}{{2^n}^2}|\tr(\widetilde{U}^\dagger U)|^2]=\frac{2^{n-m}}{2^n}=\frac{1}{2^m}$.) Therefore, access to an average-case $O(n^k)$-time learning algorithm would enable $O(n^k)$-time distinguishing of the PRF from a random function. Since the security guarantee of the PRF is negligible, and the implementation error can be made $O(1/k^\prime)$ with $k^\prime > k$, \Cref{lem:approximator} implies that no $O(1/n^k)$-distinguishing algorithm can exist, a contradiction. 
\end{proof}
Note that the ancilla-overhead reduction achieved by \Cref{thm:learning_PRF} relative to \Cref{thm:learning_PRU} is not fundamental, but rather a reflection of PRU implementation's reliance on PRFs, with additional ancilla overhead. However, we still believe the PRU proof given in \Cref{thm:learning_PRU} is of interest, as there could plausibly exist ancilla-efficient PRU constructions that are not based on PRFs.

Furthermore, we observe that the criterion of the previous learning lower-bounds almost entirely match the guarantees of the of the $\QAC^0$ unitary learning algorithm proposed by \cite{vasconcelos2025learningshallowquantumcircuits}. Crucially, the learning guarantees of the \cite{vasconcelos2025learningshallowquantumcircuits} algorithm are only proven for $\QAC^0$ with up to logarithmic ancillae, whereas our tightest learning lower-bound only holds for $\QAC^0$ with sublinear ancillae---i.e. $O(n^\delta)$ ancillae for constant $\delta >0$.\ However, granted the low-support quantum LMN conjecture \cite[Conjecture 1]{vasconcelos2025learningshallowquantumcircuits}, their learning algorithm extends to achieve quasipolynomial-time average-case learning of polynomial-sized $\QAC^0$, thus nearly matching our PRU and PRF lower-bounds.

\begin{corollary} \label{thm:learning_corollary}
    Assuming the quantum LMN low-support conjecture \cite[Conjecture 1]{vasconcelos2025learningshallowquantumcircuits}, the quasipolynomial $O(n^{\Poly\log n})$ time-complexity of the \cite{vasconcelos2025learningshallowquantumcircuits} algorithm nearly matches our super-polynomial $n^{\omega(1)}$ time lower-bound and is, thus, near-optimal.
\end{corollary}

We conclude with a couple final learning observations. First, note that our $\QAC^0$ $t$-design implementations could also be used to obtain information-theoretic lower-bounds against low-degree learning of $\QAC^0$, via the approach of \cite{chen2025information}.  Second, although an exponential worst-case learning lower-bound was previously known for $\QAC^0$ \cite{10.1145/3618260.3649722, vasconcelos2025learningshallowquantumcircuits}, this is the first \emph{average-case} lower-bound.
\subsection{Measurement-Based Random Unitaries}
\label{sec:msmt}
As a consequence of the equivalence between Model 2 and Model 3 from \Cref{sec:intro-models}, constructions of random unitaries in $\QAC^0_f$ can be recast as constant-time \textit{measurement-based} implementations of random unitaries. To better illustrate the measurement-based protocol, we provide a summary in this section. The protocol requires the ability to perform mid-circuit measurements, as well as classical parity processing of the outputs. All quantum gates can be assumed to be between neighboring pairs of qubits arranged on a two-dimensional grid. 

\begin{theorem}
    \label{thm:msmt}Using constant-depth quantum circuits with two-dimensional nearest-neighbor gates and feedforward measurement, we can exactly implement 
    \begin{itemize}
        \item Standard $\eps$-approximate $t$-designs with $n \cdot \Poly(\log nt^2/\eps) \cdot \Poly(t)$ ancillae, and quasipolynomial-secure PRUs (assuming subexponential security of LWE) with $n \Poly \log n$ ancillae. 
        \item Strong $t$-designs (to inverse exponential error) with $\Poly(n,t)$ ancillae, and strong PRUs with $\Poly(n)$ ancillae.
    \end{itemize}
\end{theorem}
\begin{proof}
Both constructions follow the same overall blueprint:
\begin{algorithm}[H]
\begin{algorithmic}[1]
\For{each layer of the glued unitary}
    \For{each local patch}
        \For{$O(1)$ repetitions}
            \State Apply a layer of single-qubit gates.
            \State Apply a tensor product of disjoint \texttt{FANOUT} gates (within the patch). 
        \EndFor
    \EndFor
\EndFor
\end{algorithmic}
\end{algorithm}
A tensor product of disjoint \texttt{FANOUT} gates is a Clifford circuit, and therefore can be implemented as a $O(\ell)$-depth circuit on a one-dimensional nearest neighbor architecture~\cite{Bravyi_2021}, where $\ell$ is the local patch size. This can be converted into a 2D nearest neighbor circuit on $O(\ell) \times O(\ell)$ qubits with constant depth using mid-circuit measurements (see \Cref{app:Clifford}). Each of the local patches (Line 2) can be placed side-by-side in a 2D nearest neighbor grid and executed in parallel: the total number of qubits required is $O(m\ell)$, where $m$ is the total number of qubits in the circuit including ancillae. After completion of the first layer, the intermediate state of the circuit will lie on all the rightmost qubits of the $O(m) \times O(\ell)$ grid of qubits—see \Cref{fig:clifford}. We can then repeat the same process with the second layer to finish executing the circuit. Taking the appropriate local patch sizes from \Cref{lem:glueddesigns}, \Cref{lem:glueprus}, \Cref{thm:fdesign}, \Cref{cor:subexp_pru} (in the latter two, which correspond to strong random unitaries, the ``patch'' is the entire circuit) we obtain constant-depth measurement based-implementations of random unitaries with the claimed number of ancillae. 
\end{proof}

\subsection{Towards \texorpdfstring{$\texttt{PARITY} \notin \QAC^0$}{}}
\label{sec:parity_lower_bound}
In this section, we show how the construction of strong unitary $t$-designs in $\QAC^0_f$ implies a connection between random unitaries and $\mathtt{PARITY} \notin \QAC^0$. Specifically, we show that any algorithm which distinguishes $\QAC^0$ circuits from Haar-random unitaries using a polynomial number of forward- and inverse-queries implies that $\texttt{PARITY} \notin \QAC^0$:
\begin{theorem}
    \label{thm:path_to_par}
    Suppose that there exists a quantum algorithm which makes $\Poly(n)$ queries to $\Phi_C$ and $\Phi_{C^\dagger}$, and distinguishes the following two cases with at least inverse exponential advantage: 
    \begin{itemize}
        \item $C$ is a $\QAC^0$ circuit.
        \item $\Phi_C = \Phi_U = U\rho U^\dagger$ for a Haar-random unitary $U$. 
    \end{itemize}
    Then, no $\QAC^0$ circuit $C^\prime$ computes the boolean function $\texttt{PARITY}(x) = \bigoplus_i x_i$ exactly. 
\end{theorem}
\label{altpath}
\Cref{thm:path_to_par} may unlock alternate paths towards proving $\texttt{PARITY} \notin \QAC^0$. For illustration, we show that having a single ``peak'' in the Pauli decomposition of a Heisenberg-evolved observable suffices to distinguish a unitary from Haar-random. In general, this requires inverse access to $U$.\   Variants of this property are implied by~\cite{Nadimpalli_2024, anshu2024computationalpowerqac0barely}—we leave details for future exploration.
To illustrate, we give an example of a candidate property which, if true, would separate $\QAC^0$ and $\QAC^0_f$. 
Suppose that for any unitary $U$ implemented by a $\QAC^0$ circuit $C$, we have $$\abs{\tr(UOU^\dagger P)/2^n}^2 \ge 1/n^{\Poly \log n}$$ for some Pauli string $P$, with high probability over a single qubit operator $O$\footnote{This is a weaker statement than the ``low-degree concentration'' conjectures of \cite{Nadimpalli_2024, vasconcelos2025learningshallowquantumcircuits}.}. Consider sampling  
\[
\rho' = (\Phi_C \Phi_O \Phi_{C^\dagger} \otimes I^{\otimes n})\ket{\mathsf{EPR}_n}
\bra{\mathsf{EPR}_n})
\]
in the Bell basis  
\[
\left\{\frac{1}{\sqrt{2^n}}(P \otimes I)\ket{\mathsf{EPR}_n} : P \in \{I, X, Y, Z\}^{\otimes n}\right\}.
\]
This yields $P^\prime$ with probability $\abs{\tr(UOU^\dagger P^\prime)/2^n}^2$. Sampling $\rho'$ twice and checking for agreement allows us to detect $P$ with $1/n^{\Poly \log n}$ probability. 

Conversely, for a Haar-random $U$, $\abs{\tr(UOU^\dagger P)/2^n}^2$ will almost always be exponentially small for all $P$. We sketch a proof: for any two non-identity \footnote{${\tr(UOU^\dagger I)} = \tr(O) = 0$, so we assume $P^\prime, Q^\prime \ne I$.} Pauli strings $P', Q'$, let $D$ be the Clifford operator such that $D P' D^\dagger = Q'$. By unitary invariance of the Haar-measure,
\begin{align*}
    \frac{1}{4^n}\mathbb{E}_{U}[|UOU^\dagger P^\prime|^2] = \frac{1}{4^n} \mathbb{E}_{U}[|\operatorname{Tr}(P^\prime UOU^\dagger)|^2]
    = \frac{1}{4^n} \mathbb{E}_{U}[|\operatorname{Tr}(P^\prime DUOU^\dagger D^\dagger)|^2]\\
    = \frac{1}{4^n} \mathbb{E}_{U}[|\operatorname{Tr}(D^\dagger P^\prime DUOU^\dagger)|^2] 
= \frac{1}{4^n} \mathbb{E}_{U}[|\operatorname{Tr}(Q^\prime UOU^\dagger)|^2] 
    =  \frac{1}{4^n}\mathbb{E}_{U}[|UOU^\dagger Q^\prime|^2].
\end{align*}
Hence, each for each non-identity $P^\prime$, $\abs{\tr(UOU^\dagger P^\prime)/2^n}^2$ has value  $1/(4^n - 1)$ in expectation.\ By concentration of measure (e.g., Lévy's Lemma, see \cite{Low_2009}), $\abs{\tr(UOU^\dagger P^\prime)/2^n}^2$ is exponentially small in $n$ except with \textit{doubly-exponentially} small probability (in $n$).\ Thus, this test distinguishes $\rho'$ from a Haar-random state, without appealing to any notion of low-degree concentration. 
\begin{proof}[Proof of \Cref{thm:path_to_par}]
    Suppose for the sake of contradiction that a $\QAC^0$ circuit $U$ exactly computes the (boolean function) $\texttt{PARITY}$, so $U$ is an $(n+a)$-qubit unitary taking $n$ input qubits and $a = \Poly(n)$ ancillae such that
    \begin{align*}
        U\ket{x_1, \dots, x_n, 0}\ket{0^a} = \ket{\texttt{PARITY}(x)}\ket{\psi_{x}} .
    \end{align*}
    For some $n + a - 1$-qubit state $\ket{\psi_x}$.\ Note that if we trace out the ancilla qubits, this channel is not in general unitary. We would like to use this to construct the \emph{unitary} implementation of $\texttt{PARITY}$. This can be done using the standard uncomputation technique: on input $\ket{x}$ we apply $U$ to $\ket{x}\ket{0^a}$ and CNOT the value $\texttt{PARITY}(x)$ into a separate register in the state $\ket{b}$.\ We then uncompute by applying $U^\dagger$ to get $\ket{x, \texttt{PARITY}(x) \oplus b}\ket{0^a}$.\ Therefore we can construct a circuit $U'$ that implements the $\texttt{PARITY}$ unitary 
    $U'\ket{x_1, \dots, x_n, b}\ket{0^a} = \texttt{PARITY}\ket{x_1, \dots, x_n, b} \ot \ket{0^a}$. Since we only used $U, U^\dagger$ and a CNOT gate, $U'$ is also a $\QAC^0$ circuit.

    Using the ability to compute \texttt{PARITY} together with the construction of strong unitary $t$-designs in \Cref{thm:fdesign}, we conclude that for any $t = \Poly(n)$, strong $\exp(-\Omega(n))$-approximate unitary $t$-designs can be implemented in $\QAC^0$.\ Therefore, for any constant $k$, there exist $\QAC^0$ circuits $C$ such that $n^{k}$ queries to $\Phi_C$ and $\Phi_{C^\dagger}$ do not suffice to distinguish $C$ from Haar-random with probability greater than $n^k \exp(-\Omega(n)) = \exp(-\Omega(n))$. This contradicts our the theorem assumptions.
\end{proof}

Furthermore, we show that if there exists an algorithm which distinguishes $\QAC^0$ circuits from Haar-random unitaries with a polynomial number of queries, then this implies that $\QAC^0$ circuits cannot compute the classical fanout function, even in the average-case. 
\begin{theorem}\label{thm:disting-implies-fanoutlb}
    Suppose that there exists a quantum algorithm which makes $\Poly(n)$ queries to $\Phi_C$ and $\Phi_{C^\dagger}$, and distinguishes the two cases with inverse polynomial error. Then no $\QAC^0$ circuit computes the (classical) \texttt{FANOUT} function $f(s, x_1, \dots, x_n) = (s, s+ x_1, \dots, x+ x_n)$ with average-case error $1/2 + \delta$, for any constant $\delta > 0$
\end{theorem}
\begin{proof}
    We now show that if a $\QAC^0$ circuit could compute the classical \texttt{FANOUT} function on average, then $\QAC^0$ circuits can approximate $t$ designs for any $t = \Poly(n)$. Suppose a $\QAC^0$ circuit $C$ satisfies 
    $$\pr_{s\in \{0,1\}, x \in \{0, 1\}^n}[C(s, x) = (s, s \oplus x_1, \dots, s \oplus x_n)] \ge \frac{1}{2} + \delta$$
    By \Cref{cor:parity_improved}, a $\Theta(\log n)$-width \texttt{FANOUT} gate can be $\Theta({1/n^{c_0}})$-approximated by a $\QAC^0$ circuit for any constant $c_0$. Consequently, a $\Theta(\log n)$-width \texttt{THRESHOLD} gate also can be $\Theta({1/\Poly(n)})$-approximated in $\QAC^0$ (\Cref{thm:tc0_power}).\ This suffices to perform standard probability amplification:\ on input $(s, x)$, apply $\Theta(\log n)$-width \texttt{FANOUT} gates to copy the input $\Theta(\log n)$ times in parallel, then apply $C$ on each, followed by using \texttt{THRESHOLD} gates in parallel to take a majority vote across the target qubits.\ Applying standard Chernoff bounds, a $1/2 + \delta$ success probability for constant $\delta$ can be boosted into a $1 - 1/\Poly(n)$ success probability with a logarithmic number of repetitions. The entire circuit, denoted $C^\prime$, is in $\QAC^0$, and the implementation error is at most inverse polynomial ( \Cref{lem:compose}). Using \Cref{lem:approximator}, this implies that 
    \begin{align}
        \pr_{s\in \zo, x \in \{0, 1\}^n}[C^\prime(s, x) = \mathtt{FANOUT}(s, x)] \ge 1 - \frac{1}{n^{c_1}} \label{eq:avgsuccfanout}
    \end{align}

    By a simple counting argument, we see that $C'$ performs well on most of the inputs. In particular, there exists a subset of inputs $S \subseteq \zo \times \zo^n$, $C'$ of size $|S| \geq \frac{2}{3}\cdot \frac{1}{2^{n+1}}$ such that $C'$ accurately computes \texttt{FANOUT} with probability at least $1-\frac{3}{n^{c_1}}$. 
    \begin{align}
        \pr[C^\prime(s, x) = \mathtt{FANOUT}(s, x)] \ge 1 - \frac{3}{n^{c_1}} && \text{for each} \ (s,x) \in S.\label{eq:fanout-good-on-DS}
    \end{align}
    To see why this is true, let $S$ be the set of all inputs that $C'$ is correct with probability at least $1-\frac{3}{n^{c_1}}$. Suppose for the sake of contradiction that $|S|< \frac{2}{3}\cdot \frac{1}{2^{n+1}}$, then 
    $\Pr_{s,x}[C'(s,x) \neq \texttt{FANOUT}(s,x)] \geq \Pr_{s,x}[(s,x) \in S] \cdot 0 + \Pr_{s,x}[(s,x) \notin S] \cdot \frac{3}{n^{c_1}} > \frac{1}{3} \cdot \frac{3}{n^{c_1}} = \frac{1}{n^{c_1}}$, contradicting \Cref{eq:avgsuccfanout}.
    
    $C$ is implemented by a unitary $U$ acting on $n+1$ inputs and $a = \Poly(n)$ ancillae. We denote its action on input  $s \in \zo, x \in \zo^n$ as follows
    \begin{align*}
        U \ket{s, x}\ket{0^a} = \alpha_{s,x} \ket{s, s \oplus x_1, \dots, s \oplus x_n}\ket{\psi_{s,x}} + \beta_{s,x} \ket{\textsf{bad}_{s,x}}.
    \end{align*}
    For some an $a$-qubit state $\ket{\psi_{s,x}}$, and an $(n+1+a)$-qubit state $\ket{\textsf{bad}_{s,x}}$ that is orthogonal to $\ket{s, s\oplus x_1, \dots, s \oplus x_n}$. Furthermore, \Cref{eq:fanout-good-on-DS} tells us that for each $(s,x) \in S$, $\abs{\alpha_{s,x}}^2 \geq 1 - \frac{3}{n^{c_1}}$ and $\abs{\beta_{s,x}}^2 \leq \frac{3}{n^{c_1}}$. Since $|S| \geq \frac{2}{3} \cdot \frac{1}{2^{n+1}}$, there exists an $x$ such that both $(0, x)$ and $(1, x)$ are in $S$. Thus, we prepare the state $\frac{1}{\sqrt{2}}(\ket{0,x} + \ket{1,x})$ and apply $U$ to it to get
    \begin{align*}
        U \cdot (\frac{1}{\sqrt{2}}(\ket{0,x} + \ket{1,x}) = \frac{1}{\sqrt{2}} \pbra{\alpha_{0,x}\ket{0, x}\ket{\psi_{0,x}} + \alpha_{1,x} \ket{1, \bar{x}} \oplus \ket{\psi_{1,x}} + \beta_{0,x} \ket{\textsf{bad}_{0,x}} + \beta_{1,x} \ket{\textsf{bad}_{1,x}}}
    \end{align*}
     Furthermore, we apply the $X$ gate to each qubit $i$ of the second register with $x_i = 1$ to get a state of the form
     \begin{align}
        \ket{\psi} := \frac{\alpha_{0,x}}{\sqrt{2}}\ket{0^{n+1}}\ket{\psi_{0,x}} + \frac{\alpha_{1,x}}{\sqrt{2}} \ket{1^{n+1}} \ket{\psi_{1,x}} + \beta \ket{\textsf{bad}}.
     \end{align}
     Where  $\beta \ket{\textsf{bad}} = \beta_{0,x} \ket{\textsf{bad}_{0,x}} + \beta_{1,x} \ket{\textsf{bad}_{1,x}}$. Furthermore, we can assume that the $\alpha_{i, x}$'s are positive real numbers since we can apply a global phase to change $\alpha_{0,x}$ and then a phase gate to change $\alpha_{1,x}$. Let $\ket{\eta}: = \frac{1}{\sqrt{2}} (\ket{0^{n+1}}\ket{\psi_{0,x}} + \ket{1^{n+1}}\ket{\psi_{1,x}})$. The state $\ket{\eta}$ is a generalization of the CAT state, and is referred to as a \emph{nekomata state} \cite{rosenthal2020boundsqac0complexityapproximating}. The fidelity of our state $\ket{\psi}$ with $\ket{\eta}$ is
     \begin{align}
         \abs{\braket{\eta | \psi}}^2 &= \abs{\frac{1}{2}(\alpha_{0,x} + \alpha_{1,x}) + \beta \braket{\eta | \textsf{bad}}}^2 \\
         &\geq \pbra{\frac{1}{2} \abs{\alpha_{0,a} + \alpha_{1,x}} - |\beta| }^2 \\
         &\geq \pbra{\sqrt{1- \frac{3}{n^{c_1}}} - 2\sqrt{\frac{3}{n^{c_1}}}}^2
         \geq 1 - 2 \sqrt{\frac{3}{n^{c_1}}}. 
     \end{align}
     Thus we have prepared a state that is a $1 - O(1/n^{c_1/2})$ close in fidelity to a nekomata state. Therefore, by \Cref{lem:neko}, the \texttt{PARITY} unitary operation can be $1 - O(1/n^{c})$-approximated in $\QAC^0$ for any constant $c$.\ Using this approximation of \texttt{PARITY} together with the construction of strong unitary $t$-designs in \Cref{thm:fdesign}, we conclude that for any $t = \Poly(n)$, strong $\exp(-\Omega(n))$-approximate unitary $t$-designs can be approximated to arbitrary inverse polynomial precision in $\QAC^0$.\ Using \Cref{lem:approximator}, we conclude that for any constants $k_1$ and $k_2$, there exist polynomial size $\QAC^0$ circuits $C$ such that $n^{k_1}$ queries to $\Phi_C$ and $\Phi_{C^\dagger}$ do not suffice to distinguish $C$ from Haar-random with probability greater than $1/n^{k_2}$.
\end{proof}
In summary, we showed that a polynomial-query distinguishing algorithm for $\QAC^0$ versus Haar random would imply that $\QAC^0$ circuits cannot compute the classical parity function \textit{exactly on all inputs}, nor the classical fanout function \textit{approximately on-average}. It is natural to ask whether it also implies that $\QAC^0$ circuit cannot compute classical parity on average. While it is known that approximate \texttt{PARITY} and \texttt{FANOUT} unitaries are equivalent up to $\QAC^0$ reductions, this doesn't translate to an equivalence about circuits that only are guaranteed to approximate \texttt{PARITY} and \texttt{FANOUT} on inputs \emph{in the computational basis}.
Rather, it is known that if a \texttt{PARITY} circuit is approximately correct on \emph{all} input states, then we can convert it to a \texttt{FANOUT} circuit that is approximately correct on all input states.

\Cref{thm:path_to_par,thm:disting-implies-fanoutlb} connect random unitaries and quantum circuit lower bounds. This connection was previously studied in the context of quantum learning of boolean functions \cite{arunachalam2021quantumlearningalgorithmsimply}, and for state synthesis \cite{chia2024quantumstatelearningimplies}. The fact that distinguishing (channels approximating) $\QAC^0$ \textit{unitaries} suffices to prove quantum circuit lower bounds for \textit{decision problems} seems to be a particular feature of $\QAC^0$, enabled by the $\QAC^0$ reduction of implementing \texttt{PARITY} to synthesizing nekomata states (\Cref{lem:neko}). 
\section{Acknowledgements}
BF thanks Dale Jacobs, Jackson Morris, Thomas Schuster, and Vishnu Iyer for helpful discussions, and HY for hosting during a visit to Columbia where this work was conducted. NP thanks Miranda Christ for helpful discussions about weak pseudorandom functions. HY thanks Harry Zhou and Soonwon Choi for helpful discussions. NP is supported by the Google PhD fellowship. FV is supported by the Paul and Daisy Soros Fellowship for New Americans, NSF Graduate Research Fellowship Grant DGE2146752, and NSF
Grant 2311733. HY is
supported by AFOSR award FA9550-23-1-0363, NSF CAREER award CCF-2144219, NSF award CCF-232993, and
the Sloan Foundation.  
\printbibliography
\appendix
\section{Appendix}  
Here, we give deferred proofs from the main body of the paper. 
\begin{lemma}
    \label{app:compose}
    Assume that $C$ $\eps_0$-approximates $U$ and $D$ $\eps_1$-approximates $V$. Then, $CD$ $(\eps_0 + \eps_1)$-approximates $UV$. 
\end{lemma}
\begin{proof}
    We want to lower bound
    \begin{equation*}
        \abs{\bra{\psi, 0^a}D^\dagger C^\dagger (UV\ket{\psi} \otimes \ket{0^a})}^2
    \end{equation*}
    Let $\ket{\phi} = V\ket{\psi} \otimes \ket{0^a}$. Using linearity of the inner product, we obtain the following chain of inequalities: 
    \begin{align*}
        \abs{\bra{\psi, 0^a}D^\dagger C^\dagger(U \otimes I)\ket{\phi}} &=  \abs{\bra{\psi, 0^a}D^\dagger\ket{\phi}} +  \abs{\bra{\psi, 0^a}D^\dagger(C^\dagger(U \otimes I) - I)\ket{\phi}} \\
        &\ge \sqrt{(1 - \eps_1)} +  \abs{\bra{\psi, 0^a}D^\dagger(C^\dagger(U \otimes I) - I)\ket{\phi}} \\
        &\ge \sqrt{(1 - \eps_1)} + \norm{C^\dagger(U \otimes I)\ket{\phi} - \ket{\phi}}_2\\
        &= \sqrt{(1 - \eps_1)} + \sqrt{1 - \eps_0} - 1  \\ 
        \implies \abs{\bra{\psi, 0^a}D^\dagger C^\dagger (UV\ket{\psi} \otimes \ket{0^a})}^2 &\ge  \left(\sqrt{(1 - \eps_1)} + \sqrt{1 - \eps_0} - 1\right)^2  \ge 1 - \eps_0 - \eps_1
    \end{align*}
    from the first to the second line, we used that $D$ $\eps_1$-approximates $V$, and from the third to fourth line that $C$ $\eps_0$-approximates $U$. 
\end{proof}
\begin{lemma}
    \label{app:Clifford}
    Let $C$ be any $n$-qubit Clifford Circuit.\ $C$ is implementable by a $\QAC^0_f$ circuit with $O(n^3)$ ancillae. Alternatively, $C$ is implementable by a constant-depth circuit with intermediate measurements and parity corrections, using $O(n^2)$ ancillae (and classical memory), and 2D nearest-neighbor gates. 
\end{lemma}
\begin{proof}
    This lemma is nearly identical to many measurement-based constructions for Clifford circuits present in the literature \cite{jozsa2005introductionmeasurementbasedquantum, Buhrman_2024, parham2025quantumcircuitlowerbounds}.\ First, it is well-known that any $n$-qubit Clifford circuit $C$ can be decomposed into a $O(n)$ depth circuit $C_kC_{k-1}\dots C_2C_1$ consisting of 1D nearest-neighbor two-qubit Clifford gates \cite{Bravyi_2021} (see \Cref{fig:clifford}).
    \begin{figure}[ht]
    \centering
    \includegraphics[scale=0.5]{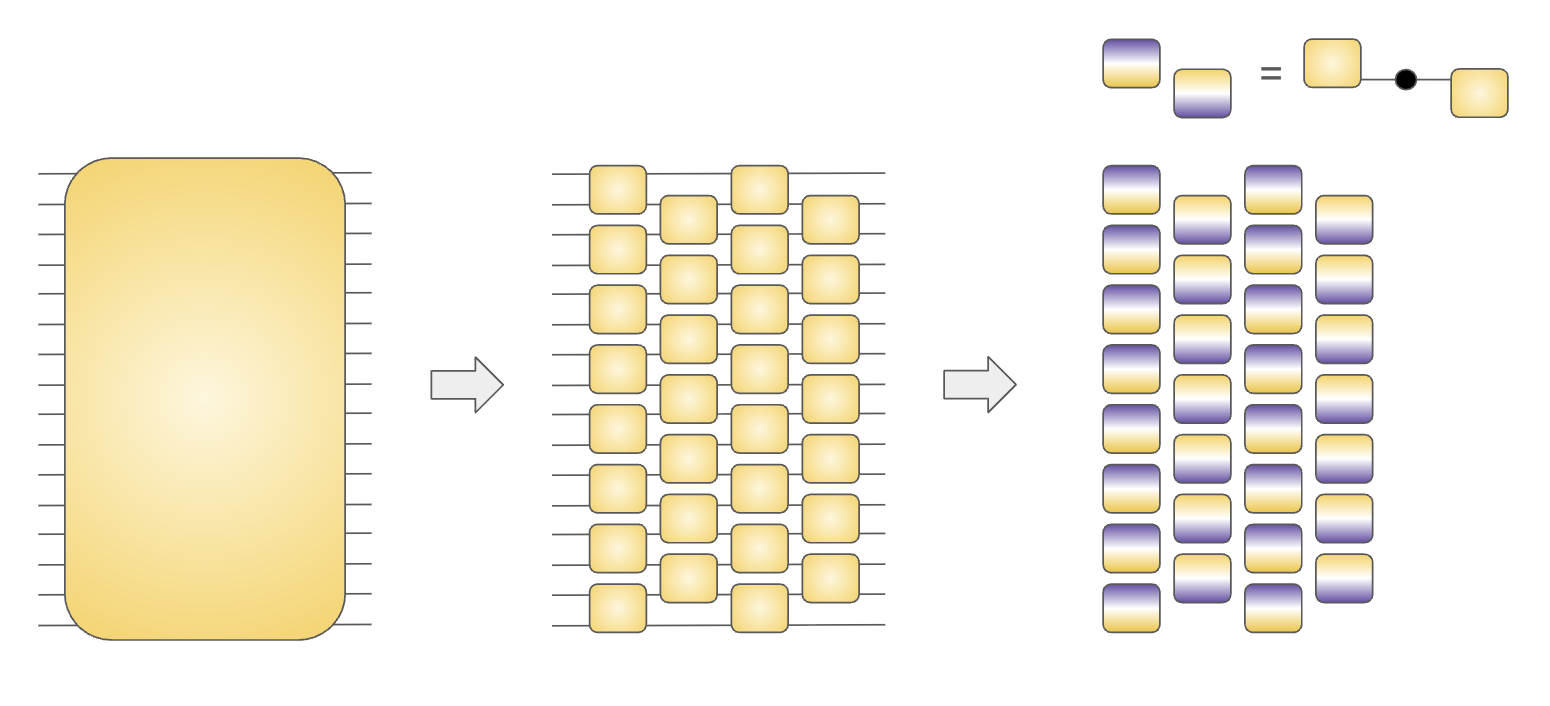}
    \caption{Converting any circuit into a linear-depth nearest neighbor circuit, then using Bell state preparation and measurement (with ancillae, denoted by the multi-colored boxes) to parallelize the circuit with a 2D nearest-neighbor architecture. Time flows left to right. The output state after Bell measurements (up to corrections) lies on the right edge of the grid. For Clifford circuits, the result of the measurement can be corrected efficiently.}
    \label{fig:clifford}
    \end{figure}
    Using standard Bell state teleportation gadgets (see \cite{Buhrman_2024}) to trade time for space, we obtain a $O(1)$ depth circuit which implements the operator 
    \begin{equation*}
         \sum_{m} (X_1^{c_1 \cdot m}Z_1^{c_2 \cdot m}\dots  X_n^{c_{2n-1} \cdot m}Z_n^{c_{2n} \cdot m}C)_\mathsf{A} \otimes \ket{m}\bra{m}_\mathsf{R}
    \end{equation*}
    where $\mathsf{R}$ is an ancilla register with $O(n^2)$ qubits, and the $c_i$'s are bits which depend on the decomposition of $C$ into two qubit unitaries.\ Then, we apply a \texttt{FANOUT} gate controlled on the qubits in $\mathsf{R}$ to $2n$ additional target registers, each with the same size as $\mathsf{R}$: 
    \begin{equation*}
        \sum_{m} (X_1^{c_1 \cdot m}Z_1^{c_2 \cdot m}\dots  X_n^{c_{2n-1} \cdot m}Z_n^{c_{2n} \cdot m}C)_\mathsf{A} \otimes \ket{m}\bra{m}_\mathsf{R} \otimes \bigotimes_{i=1}^{2n} \ket{m}\bra{m}_\mathsf{R_i}
    \end{equation*}
    (If using mid-circuit measurement, we simply measure $\mathsf{R}$).  
    We then coherently compute all $2n$ inner products, doing one inner product in each register, by applying a \texttt{PARITY} = $H^{\otimes O(n^2)} \cdot \texttt{FANOUT} \cdot H^{\otimes O(n^2)}$ gate to each copied register, possibly adding $X$ gates to some of the qubits to compute the appropriate $\mathbb{F}_2$-inner product. If using mid-circuit measurement, these inner products can be computed classically using the obtained measurement outcome. 
    \begin{equation*}
        \sum_{m} (X_1^{c_1 \cdot m}Z_1^{c_2 \cdot m}\dots  X_n^{c_{2n-1} \cdot m}Z_n^{c_{2n} \cdot m}C)_\mathsf{A} \otimes \ket{m}\bra{m}_\mathsf{R} \otimes \bigotimes_{i=1}^{2n} (\ket{m}\bra{m}_\mathsf{R_i} \otimes \ket{c_i \cdot m}\bra{c_i \cdot m}_\mathsf{S_i})
    \end{equation*}
    Then, we apply the appropriate controlled-$X$ and controlled-$Z$ gates from the qubits storing the values of these inner products to the appropriate target qubits in $\mathsf{A}$ to cancel out the extraneous Pauli operators:
    \begin{equation*}
        \sum_{m} C_\mathsf{A} \otimes \ket{m}\bra{m}_\mathsf{R} \otimes \bigotimes_{i=1}^{2n} (\ket{m}\bra{m}_\mathsf{R_i} \otimes \ket{c_i \cdot m}\bra{c_i \cdot m}_\mathsf{S_i})
    \end{equation*}
    We can uncompute the previous two steps to remove the garbage in the $\mathsf{R_i}$ and $\mathsf{S_i}$ registers. Overall, this results in a $O(1)$ depth $\QAC^0_f$ circuit which implements $C$, using $2n \cdot O(n^2) = O(n^3)$ ancillae, or $O(n^2)$ ancillae using mid-circuit measurement. 
\end{proof}
\end{document}